\documentclass[oneside]{amsart}

\usepackage[utf8]{inputenc} 
\usepackage[T1]{fontenc}    
\usepackage{hyperref}       
\usepackage{url}            
\usepackage{booktabs}       
\usepackage{amsfonts}       
\usepackage{microtype}      

\usepackage{graphicx}
\usepackage{xcolor}
\usepackage{epstopdf}
\usepackage{caption}

\usepackage{subfigure}
\usepackage{indentfirst}
\usepackage{amsmath}

\usepackage{multirow}
\usepackage{threeparttable} 
\usepackage{caption}
\usepackage{enumerate}
\usepackage{bm}
\usepackage{listings}
\usepackage{float}

\newtheorem{definition}{Definition}
\newtheorem{theorem}{Theorem}
\newtheorem{lemma}{Lemma}

\def\e{\bm{e}}
\def\f{\bm{f}}
\def\g{\bm{g}}
\def\h{\bm{h}}
\def\m{\bm{m}}
\def\p{\bm{p}}
\def\q{\bm{q}}
\def\x{\bm{x}}

\usepackage[ruled,linesnumbered]{algorithm2e}
\usepackage{makecell}

\usepackage{cleveref}
\crefname{equation}{}{}
\crefname{figure}{Fig.}{Fig.}
\crefname{table}{Table}{Table}

\title[Advantages of a semi-implicit scheme for LLG equation]{Advantages of a semi-implicit scheme over a fully implicit scheme for Landau-Lifshitz-Gilbert equation}
\author[Y. Sun]{Yifei Sun}
\address{School of Mathematical Sciences\\ Soochow University\\ Suzhou\\ China.}
\email{1019227379@qq.com}

\author[J. Chen]{Jingrun Chen}
\address{School of Mathematical Sciences and Mathematical Center for Interdisciplinary Research\\ Soochow University\\ Suzhou\\ China.}
\email{jingrunchen@suda.edu.cn}

\author[R. Du]{Rui Du}
\address{School of Mathematical Sciences and Mathematical Center for Interdisciplinary Research\\ Soochow University\\ Suzhou\\ China.}
\email{durui@suda.edu.cn (Corresponding author)}

\author[C. Wang]{Cheng Wang}
\address{Mathematics Department\\ University of Massachusetts\\ North Dartmouth\\ MA 02747\\ USA.}
\email{cwang1@umassd.edu}

\begin{document}


\date{\today}


\begin{abstract}
Magnetization dynamics in magnetic materials is modeled by the Landau-Lifshitz-Gilbert (LLG) equation, which is a nonlinear system of partial differential equations. In the LLG equation, the length of magnetization is conserved and the system energy is dissipative. Implicit and semi-implicit schemes have often been used in micromagnetics simulations due to their unconditional numerical stability. In more details, implicit schemes preserve the properties of the LLG equation, but solve a nonlinear system of equations per time step. In contrast, semi-implicit schemes only solve a linear system of equations, while additional operations are needed to preserve the length of magnetization. It still remains unclear which one shall be used if both implicit and semi-implicit schemes are available. In this work, using the implicit Crank-Nicolson (ICN) scheme as a benchmark, we propose to make this implicit scheme semi-implicit. It can be proved that both schemes are second-order accurate in space and time. For the unique solvability of nonlinear systems of equations in the ICN scheme, we require that the temporal step size scales quadratically with the spatial mesh size. It is numerically verified that the convergence of the nonlinear solver becomes slower for larger temporal step size and multiple magnetization profiles are obtained for different initial guesses. The linear systems of equations in the semi-implicit CN (SICN) scheme are unconditionally uniquely solvable, and the condition that the temporal step size scales linearly with the spatial mesh size is needed in the convergence of the SICN scheme. In terms of numerical efficiency, the SICN scheme achieves the same accuracy as the ICN scheme with less computational time. Based on these results, we conclude that a semi-implicit scheme is superior to its implicit analog both theoretically and numerically, and we recommend the semi-implicit scheme in micromagnetics simulations if both methods are available.
\end{abstract}

\maketitle

\section{Introduction}
In most solid materials, local magnetic orders posed by electrons do not generate magnetization macroscopically, because of the symmetry. Ferromagnets, however, spontaneously break the symmetry, and yield a macroscopic magnetization. Ferromagnets exhibit stable binary configurations, which makes them novel materials for data storage. Recent experimental advances have demonstrated the feasibility of effective and precise control of ferromagnetic structures by applying external fields \cite{zutic-2004}. Phenomenologically, magnetization dynamics is modeled by the Landau-Lifshitz-Gilbert (LLG) equation \cite{Gilbert1955A, landau1992theory}. Over the past decades, extensive numerical simulations have continued to overcome the difficulties associated with the nonlinearity of the model and the non-convex constraint on the magnetization (see \cite{Cimrak2008, Kruzik2006} for reviews and references therein). 

For temporal discretization, there have been explicit schemes \cite{ALOUGES2006}, implicit schemes \cite{Bartels2006,Fuwa2011}, and semi-implicit schemes \cite{Boscarino2016, Cimrak2005, Weinan2000, Gao2014, Lewis2003, li2020two, Wang2001, XIE2020109104}. Explicit schemes usually lead to a severe time step constraint for the numerical stability reason, while implicit and semi-implicit schemes are verified to be unconditionally stable, and thus are commonly used. Meanwhile, implicit and semi-implicit schemes have their advantages and disadvantages, respectively. On the one hand, implicit schemes preserve the length of magnetization and inherit the energy dissipation property automatically, while the semi-implicit schemes do not preserve the length of magnetization, which is compensated by a projection step. Moreover, the property of energy dissipation is not theoretically preserved for the semi-implicit schemes. These issues may affect the numerical accuracy in the long time simulation. On the other hand, semi-implicit schemes only require one linear system solver at each time step, so that its efficiency is obvious. In contrast, a nonlinear system has to be solved for the implicit schemes at each time step. 
In fact, a fixed-point iteration method, such as standard simple linearization or Newton's method, is employed to solve nonlinear systems of equations. Several iterations are needed at each time step to achieve iteration convergence for the implicit schemes. Due to the non-convex nature of the problem, the uniqueness of the numerical solution of the implicit schemes could only be established when $k=\mathcal{O}(h^2)$, with $k$ being the temporal step size and $h$ being the spatial mesh size \cite{Fuwa2011}. For the semi-implicit schemes, the unique solvability analysis could be theoretically justified for any $k, h > 0$. It is worth mentioning that no matter which scheme, the linear or linearized system is non-symmetric and is solved by GMRES in this work.

A natural question arises when solving the LLG equation: which one shall be preferred if both implicit and semi-implicit schemes are available? The answer to this question remains unclear in the existing literature. In this paper, we take the implicit Crank-Nicolson (ICN) scheme as an example, and try to make this implicit scheme semi-implicit. These two schemes are second-order accurate in both space and time. The numerical solution to a nonlinear system of equations in the ICN scheme is unique under a severe time step constraint $k=\mathcal{O}(h^2)$, while linear systems of equations in the semi-implicit Crank-Nicolson (SICN) scheme preserve the unconditionally unique solvability. To ensure the convergence of the SICN scheme, a much milder constraint, $k=\mathcal{O}(h)$, is needed. Moreover, the semi-implicit schemes are much more advantageous in terms of numerical efficiency, since only a linear system solver is needed at each time step. If both schemes work, the SICN scheme can reduce the computational time without loss of accuracy and shall be preferred in micromagnetics simulations.

The rest of the paper is organized as follows. In Section~\ref{sec:method}, we present the ICN and SICN schemes, and discuss their properties in details, such as conservation of magnetization length, energy dissipation, solvability, and convergence. In Section~\ref{sec:result}, accuracy and efficiency tests are conducted for both schemes, \textcolor{red}{and the limitations of the nonlinear solver in the ICN scheme is demonstrated.} Micromagnetics simulations, including different stable structures and a benchmark problem from the National Institute of Science and Technology (NIST), are performed with the SICN scheme in Section \ref{sec:micromagnetics}. Finally, some conclusions are made in Section \ref{sec:conclusion}.

\section{The second-order schemes}\label{sec:method}
\subsection{Model}
The LLG equation \cite{Gilbert1955A, landau1992theory} is a phenomenological model to describe the magnetization dynamics in a ferromagnetic material, which in the dimensionless form reads as
\begin{equation}\label{LLG}
	\bm{m}_t = -\bm{m}\times \bm{h}_{\text{eff}} - \alpha \bm{m} \times (\bm{m} \times \bm{h}_{\text{eff}}) , 
\end{equation}
coupled with homogeneous Neumann boundary condition 
\begin{equation}\label{eq7}
	\frac{\partial \bm{m}}{\partial \nu} \Big{\vert}_{\partial\Omega} = 0.
\end{equation}
Here $ \Omega $ is a bounded domain occupied by the ferromagnetic material, the magnetization $ \bm{m}: \Omega \subset \mathbb{R}^d \to S^2, d=1,2,3$ is a three dimensional vector field with a point-wise constraint $ |\bm{m}| = 1 $, and $ \nu $ is the unit outward normal vector along $\partial\Omega$. On the right-hand side of \eqref{LLG}, the first term is the gyromagnetic term, while the second term stands for the damping term with a dimensionless parameter $ \alpha > 0  $.

The effective field ${\bm h}_{\text{eff}}$ is obtained by taking the variation of the Gibbs free energy of the magnetic body with respect to $\m$. The free energy $F[\m]$ includes the exchange energy, the anisotropy energy, the magnetostatic energy, and the Zeeman energy. For a uniaxial material, the following form is taken  
\begin{equation*}\label{LL-Energy}
F[\m] = \frac {\mu_0 M_s^2}{2}\int_\Omega \left( \epsilon|\nabla\m|^2 +
Q(m_2^2 + m_3^2)
- \h_s\cdot\m -2\h_e\cdot\m \right)\mathrm{d}\x. 
\end{equation*}
Accordingly, the effective field $ \bm{h}_\text{eff} $ consists of the exchange field, the anisotropy field, the demagnetization or stray field $\bm{h}_s$, and the applied external field $\bm{h}_e$, i.e.,  
\begin{equation}\label{heff}
	\bm{h}_{\text{eff}} = \epsilon \Delta \bm{m} - Q(m_2 \bm{e}_2 + m_3 \bm{e}_3) + \bm{h}_s + \bm{h}_e.
\end{equation}
The dimensionless parameters are given by: $\epsilon = 2 C_{ex}/(\mu_0 M_s^2 L^2) $, $ Q = K_u/(\mu_0 M_s^2) $, $C_{\text{ex}}$ is the exchange constant, $K_u$ is the anisotropy constant, $ L $ is the diameter of the ferromagnetic body,  $ \mu_0 $ is the permeability of vacuum, and $ M_s $ stands for the saturation magnetization. The two unit vectors turn out to be $ \bm{e}_2 = (0,1,0)^T, \bm{e}_3 = (0,0,1)^T $, and $ \Delta $ denotes the standard Laplacian operator. Typical values of the physical parameters for Permalloy are included in \cref{tab1}. 
\begin{table}[H]
\centering
\caption{Typical values of the physical parameters for Permalloy, which is an alloy of Nickel (80\%) and Iron (20\%) frequently used in magnetic storage devices.}
\label{tab1}
\begin{tabular}{c|c}
\Xhline{1.5pt}
\multicolumn{2}{c}{\textbf{Physical Parameters for Permalloy}} \\ \hline
\rule{0pt}{10pt} $ K_u $ & $ 5.0 \times 10^2 J/m^3 $ \\ \hline
\rule{0pt}{10pt}$ C_{\text{ex}} $ & $ 1.3\times 10^{-11} J/m $ \\ \hline
\rule{0pt}{10pt}$ M_s $ & $ 8.0\times 10^5 A/m $ \\ \hline
\rule{0pt}{10pt}$ \mu_0 $ & $ 4\pi \times 10^{-7} N/A^2 $ \\ \hline
\rule{0pt}{10pt}$ \alpha $ & 0.1 \\ \Xhline{1.5pt}
\end{tabular}
\end{table}

The stray field $\h_s$ takes the form
\begin{equation}\label{eqq-5}
{\h}_{\text{s}}=-\nabla \int_{\Omega} \nabla N({\bm x}-{\bm y})\cdot {\bm m}({\bm y})\,d{\bm y},
\end{equation}
where $N({\bm x})=-\frac{1}{4\pi |{\bm x}|}$ is the Newtonian potential.

For the sake of simplicity, we define
\begin{equation*}
	\bm{f} = -Q(m_2 \bm{e}_2+ m_3 \bm{e}_3)+ \bm{h}_s +\bm{h}_e,
\end{equation*}
and rewrite \eqref{LLG} as 
\begin{equation}\label{eq10}
	\bm{m}_t = -\bm{m} \times (\epsilon \Delta \bm{m} + \bm{f} ) - \alpha \bm{m} \times \left(\bm{m} \times (\epsilon \Delta \bm{m} + \bm{f} )\right).
\end{equation}

\subsection{The Crank-Nicolson scheme}\label{discret}

Denote the temporal stepsize by $ k $, and $ t^n = nk $, $ n\le \lfloor \frac{T}{k} \rfloor $ with $  T $ the final time. Denote the spatial meshsize by $ h $, the standard second-order centered difference for Laplacian operator by $ \Delta_h $. We use the finite difference method to approximate the spatial derivatives in \eqref{eq10}

\[
\begin{aligned}
\Delta_{h}\bm{m}_{i, j, l} =&\frac{\bm{m}_{i+1, j, l}-2\bm{m}_{i, j, l}+\bm{m}_{i-1, j, l}}{h^2} \\
&+\frac{\bm{m}_{i, j+1, l}-2\bm{m}_{i, j, l}+\bm{m}_{i, j-1, l}}{h^2} \\
&+\frac{\bm{m}_{i, j, l+1}-2\bm{m}_{i, j, l}+\bm{m}_{i, j, l-1}}{h^2},
\end{aligned}
\]
where $\bm{m}_{i, j, l}$ stands for a numerical approximation for $\bm{m}$ at a cell-centered mesh point $\left(\left(i-\frac{1}{2}\right) h,\left(j-\frac{1}{2}\right) h,\left(l-\frac{1}{2}\right) h\right)$.
For the Neumann boundary condition \eqref{eq7}, a second-order discretization yields

\begin{equation}\label{neumann}
	\begin{array}{lll}
\bm{m}_{0, j, l}=\bm{m}_{1, j, l}, & \bm{m}_{N_x, j, l}=\bm{m}_{N_x+1, j, l}, & j=1, \cdots, N_y, l=1, \cdots, N_z, \\
\bm{m}_{i, 0, l}=\bm{m}_{i, 1, l}, & \bm{m}_{i, N_y, l}=\bm{m}_{i, N_y+1, l}, & i=1, \cdots, N_x, l=1, \cdots, N_z, \\
\bm{m}_{i, j, 0}=\bm{m}_{i, j, 1}, & \bm{m}_{i, j, N_z}=\bm{m}_{i, j, N_z+1}, & i=1, \cdots, N_x, j=1, \cdots, N_y.
\end{array}
\end{equation}

For brevity, we use $\bm{m}_h$ to denote the grid function of $\bm{m}$ over the uniform grids. For the temporal discretization, we employ the Crank-Nicolson algorithm to approximate the temporal derivative
\begin{equation*}
	\frac{\partial}{\partial t} \bm{m}_h^{n+\frac{1}{2}} \approx \frac{\bm{m}_h^{n+1}-\bm{m}_h^n}{k},
\end{equation*}
which gives the following ICN scheme:
\begin{equation}\label{eq12}
	\begin{aligned}
		\frac{\bm{m}_h^{n+1}-\bm{m}_h^n}{k} = &-\bm{m}_h^{n+\frac{1}{2}} \times (\epsilon \Delta_h \bm{m}_h^{n+\frac{1}{2}} + \bm{f}_h^{n+\frac{1}{2}})\\
		& - \alpha \bm{m}_h^{n+\frac{1}{2}} \times \left (\bm{m}_h^{n+\frac{1}{2}} \times (\epsilon \Delta_h \bm{m}_h^{n+\frac{1}{2}} + \bm{f}_h^{n+\frac{1}{2}} ) \right).
	\end{aligned}
\end{equation}
Here $\displaystyle \bm{m}_h^{n+\frac{1}{2}} = \frac{\bm{m}_h^{n+1}+ \bm{m}_h^{n}}{2} $ , $\displaystyle \bm{f}_h^{n+\frac{1}{2}} = \frac{\bm{f}_h^{n+1}+ \bm{f}_h^{n}}{2} $.

To ease the description, we simplify \eqref{eq12} as 
\begin{equation}\label{eq_length}
\frac{\bm{m}_h^{n+1}-\bm{m}_h^n}{k} = -\bm{m}_h^{n+\frac{1}{2}} \times \Delta_h \bm{m}_h^{n+\frac{1}{2}} - \alpha \bm{m}_h^{n+\frac{1}{2}}\times (\bm{m}_h^{n+\frac{1}{2}}\times \Delta_h \bm{m}_h^{n+\frac{1}{2}}).
\end{equation}
We can solve the nonlinear systems of equations in \eqref{eq_length} by the following three different strategies. 

\begin{itemize}
\item[(a)] Explicit iteration
\begin{equation*}\label{iter_ex}
\frac{\m_h^{n,\ell+1} - \m_h^n}{k} = - \m_h^{n+\frac{1}{2},\ell} \times \Delta_h \m_h^{n+\frac{1}{2},\ell} - \alpha \m_h^{n+\frac{1}{2},\ell} \times \left( \m_h^{n+\frac{1}{2},\ell} \times \Delta_h \m_h^{n+\frac{1}{2},\ell} \right), 
\end{equation*}
where $\displaystyle \m_h^{n+\frac{1}{2},\ell} = \frac{\m_h^{n,\ell} + \m_h^n}{2}$. It will be prved later that $\lim_{\ell \to +\infty} \frac{\m_h^{n,\ell} + \m_h^n}{2} $ exists for every $n$, and thus we can set $\m_h^{n+1} = \lim_{\ell \to +\infty} \m_h^{n,\ell}$.
\item[(b)] Semi-implicit iteration
\begin{equation*}\label{iter_im}
\frac{\m_h^{n,\ell+1} - \m_h^n}{k} = - \m_h^{n+\frac{1}{2},\ell} \times \Delta_h \m_h^{n+\frac{1}{2},\ell+1} - \alpha \m_h^{n+\frac{1}{2},\ell} \times \left( \m_h^{n+\frac{1}{2},\ell} \times \Delta_h \m_h^{n+\frac{1}{2},\ell+1} \right).
\end{equation*}
A linear system of equations needs to be solved at each iteration and its unique solvability can be  similarly proved. 
We expect $\m_h^{n+1} = \lim_{\ell \to +\infty} \m_h^{n,\ell} $ under mild conditions.

\item[(c)] Newton's iteration
\begin{equation}\label{eqn:Newton}
\mathcal{N}(\bm{m}_h^{n+1}) = \bm{\phi}^n,  
\end{equation}
where
\[ 
\begin{aligned} 
\mathcal{N}(\m) =& \m +k \frac{\m+\m_h^n}{2} \times \Delta_h\frac{\m+\m_h^n}{2} \\ &+ \alpha k \frac{\m+\m_h^n}{2} \times \left( \frac{\m+\m_h^n}{2} \times \Delta_h \frac{\m+\m_h^n}{2} \right),
\end{aligned} 
\]
and
\[
\bm{\phi}^n=\bm{m}_h^n.
\]
The iteration algorithm \eqref{eqn:Newton} is obtained by expanding \eqref{eq_length} and combing the corresponding terms.

The details of Newton's method are given in Algorithm \ref{algorithm} with
\begin{align*} 
\mathcal{D}(\m)\delta\m  
= & \delta\m - k \Delta_h\frac{\m+\m_h^n}{2}\times \frac{\delta\m}{2}+k\frac{\m+\m_h^n}{2} \times \Delta_h\frac{\delta\m}{2} \\
& -2\alpha k \frac{\m+\m_h^n}{2} \times \left( \Delta_h \frac{\m+\m_h^n}{2} \times \frac{\delta\m}{2} \right) \\
& +\alpha k \Delta_h \frac{\m+\m_h^n}{2} \times \left(\frac{\m+\m_h^n}{2} \times \frac{\delta\m}{2} \right) \\
& +\alpha k \frac{\m+\m_h^n}{2} \times \left( \frac{\m+\m_h^n}{2} \times \Delta_h \frac{\delta\m}{2} \right),
\end{align*}
and
\begin{align*}
\mathcal{D}(\m) = & 	\bm{I} - \frac{k}{4}\left(\m+\m_h^n\right)\times + \frac{k}{4}\left(\m+\m_h^n\right)\times\Delta_h \\
& -\frac{\alpha k}{4} \left(\m+\m_h^n\right) \times \left[ \Delta_h \left(\m+\m_h^n\right) \times \right] \\
& +\frac{\alpha k}{8} \Delta_h \left(\m+\m_h^n\right) \times \left[ \left(\m+\m_h^n\right) \times \right] \\
& +\frac{\alpha k}{8} \left(\m+\m_h^n\right) \times \left[ \left(\m+\m_h^n\right) \times \Delta_h \right] . 
\end{align*}
\begin{algorithm}[H]
	\caption{Newton's Method}\label{algorithm}
	\KwIn{$ \bm{m}_h^n $}
	\KwOut{$ \bm{m}_h^{n+1} $}
	\textbf{Initialization}: $\m^* \leftarrow \m_h^n$\;
	\qquad \qquad \qquad\qquad $ |\delta\bm{m}| = 1$, $\text{iter} = 0$\;
	\textbf{Calculate}: $\bm{\phi}^n$\;
	
	\While{$ |\delta\bm{m}| >$ tol \textbf{and} $\text{iter} < \text{MaxIter}$}{
		Calculate: $\mathcal{N}(\m^*), \mathcal{D}(\m^*)$\; $\mathcal{D}(\m^*) \delta\bm{m} = \bm{\phi}^n - \mathcal{N}(\m^*)$\;
		$\m^* \leftarrow \m^* + \lambda \delta\bm{m}$\; }
	
	$\m_h^{n+1} \leftarrow \m^*$.
\end{algorithm}

\end{itemize}

Strategy (a) will be used to show the uniqueness of a solution to the nonlinear system of equations \eqref{eq_length}; see Theorem \ref{thm:ICNsolvability}. Strategy (b) is a simple fixed-point iteration and Strategy (c) (the Newton's method) can be viewed as the special fixed-point iteration. In particular, Strategy (c) has the best efficiency. Therefore, we will apply Newton's method in Section \ref{sec:result}.

To proceed, we introduce some notations of grid functions.
\begin{definition}[Discrete inner product]
	For grid functions $\bm{f}_h$ and $\bm{g}_h$ over the uniform numerical grid, we define
	\[ \langle \f_h,\g_h \rangle  = h^d 
	\sum_{\mathcal{I} \in \Lambda_d } \f_{\mathcal{I}} \cdot \g_{\mathcal{I}} , \]
where $\Lambda_d$ is the index set and $\mathcal{I}$ is the index which closely depends on the dimension d.

\end{definition}

\begin{definition}
	For the grid function $\f_h$, the average of summation is defined as 
	\[\bar{\f_h} = h^d \sum_{\mathcal{I} \in \Lambda_d} \f_{\mathcal{I}}.\]
\end{definition}

\begin{definition}
	For the grid function $\f_h$ with $\bar{\f_h} =0 $ and the homogeneous Neumann boundary condition, the discrete $H_h^{-1}$ and $ H_h^1$ norms are introduced as 
	\[
	\| \f_h \|_{-1} = \langle (-\Delta_h)^{-1} \f_h, \f_h \rangle^{1/2},
	\] 
	\[
	\| \f_h\|_{H^1_h} = (\|\f_h\|_2^2 + \|\nabla_h \f_h\|_2^2)^{1/2} .
	\]
\end{definition}

\begin{lemma}[{\cite[Lemma 1]{Chen2021}}]\label{dislemma} For grid functions $\f_h$ and $\g_h$ over the uniform numerical grid, we have
\begin{align*}
\left\|\nabla_{h}\left(\bm{f}_{h} \times \bm{g}_{h}\right)\right\|_{2}^{2} & \leq \mathcal{C}\left(\left\|\bm{f}_{h}\right\|_{\infty}^{2}\left\|\nabla_{h} \bm{g}_{h}\right\|_{2}^{2}+\left\|\bm{g}_{h}\right\|_{\infty}^{2}\left\|\nabla_{h} \bm{f}_{h}\right\|_{2}^{2}\right) , \\
\left\langle\left(\bm{f}_{h} \times \Delta_{h} \bm{g}_{h}\right) \times \bm{f}_{h}, \bm{g}_{h}\right\rangle &=\left\langle\bm{f}_{h} \times\left(\bm{g}_{h} \times \bm{f}_{h}\right), \Delta_{h} \bm{g}_{h}\right\rangle ,\\
\left\langle\bm{f}_{h} \times\left(\bm{f}_{h} \times \bm{g}_{h}\right), \bm{g}_{h}\right\rangle &=-\left\|\bm{f}_{h} \times \bm{g}_{h}\right\|_{2}^{2} .
\end{align*}
\end{lemma}

In the following part, we focus on the investigation of the theoretical properties for the ICN scheme.

\noindent \textit{Length conservation of the ICN scheme}

Taking a vector inner product on both sides of \eqref{eq_length} with $\displaystyle  \frac{\bm{m}_h^{n+1}+ \bm{m}_h^{n}}{2}=\bm{m}_h^{n+\frac{1}{2}} $, we obtain
\begin{align*}
	\frac{\bm{m}_h^{n+1}-\bm{m}_h^n}{k} \cdot \frac{\bm{m}_h^{n+1}+ \bm{m}_h^{n}}{2}
	= & \bm{m}_h^{n+\frac{1}{2}} \times \Delta_h \bm{m}_h^{n+\frac{1}{2}} \cdot \bm{m}_h^{n+\frac{1}{2}} \\
	 & - \alpha \bm{m}_h^{n+\frac{1}{2}}\times (\bm{m}_h^{n+\frac{1}{2}}\times \Delta_h \bm{m}_h^{n+\frac{1}{2}}) \cdot \bm{m}_h^{n+\frac{1}{2}} \\
	= & 0 .
\end{align*}
Henceforth, the following identity is valid at a point-wise level: 
\[{|\bm{m}_h^{n+1}|}^2 = {|\bm{m}_h^n|}^2 .\]
Of course, it is straightforward to verify that this equality also holds true for \eqref{eq12}.

\noindent \textit{Energy dissipation of the ICN scheme}

Here we consider the 1D scheme, which can be easily extended to the 3D case. Let $E_h(\m_h^n)$ be the discrete energy defined as 
\[
\begin{aligned}
E_{h}\left(\m_h^{n}\right) 
=\frac{1}{h} \sum_{i=1}^{N_{x}}\left|\m_{i+1}^{n}-\m_{i}^{n}\right|^{2}.
\end{aligned}
\]
The difference of energy between the two successive steps turns out to be 
\begin{equation}\label{enerdif}
\begin{aligned}
E_h\left(\m_h^{n+1}\right)-E_h\left(\m_h^{n}\right) =&\frac{1}{h} \sum_{i=1}^{N x}\left(\left|\m_{i+1}^{n+1}-\m_{i}^{n+1}\right|^{2}-\left|\m_{i+1}^{n}-\m_{i}^{n}\right|^{2}\right) \\
= &\frac{1}{h} \sum_{i=1}^{N x}\left(\left|\m_{i+1}^{n+1}\right|^{2}-\left|\m_{i+1}^{n}\right|^{2}+\left|\m_{i}^{n+1}\right|^{2}-\left|\m_{i}^{n}\right|^{2} \right. \\ 
& \left.\quad\qquad -2 \m_{i}^{n+1} \cdot \m_{i+1}^{n+1}+2 \m_{i}^{n} \cdot \m_{i+1}^{n}\right) \\
= & - \frac{2}{h} \sum_{i=1}^{N_x}\left(\m_{i}^{n+1} \cdot \m_{i+1}^{n+1}- \m_{i}^{n} \cdot \m_{i+1}^{n}\right).
\end{aligned}
\end{equation}
Taking a discrete inner product on both sides of \eqref{eq_length} with $\Delta_h\frac{\m_h^{n+1}+\m_h^n}{2}=\Delta_h\m_h^{n+\frac{1}{2}}$, we obtain
\begin{equation}\label{hold1} 
\begin{aligned} 
  & 
\langle \frac{\m_h^{n+1}-\m_h^n}{k} , \Delta_h\frac{\m_h^{n+1}+\m_h^n}{2} \rangle 
\\
  = & \langle -\m_h^{n+\frac{1}{2}} \times \Delta_h \m_h^{n+\frac{1}{2}} - \alpha \m_h^{n+\frac{1}{2}}\times (\m_h^{n+\frac{1}{2}}\times \Delta_h\m_h^{n+\frac{1}{2}})
, \Delta_h \m_h^{n+\frac{1}{2}}\rangle. 
\end{aligned} 
\end{equation}
A direct calculation gives
\begin{equation}\label{hold2}
	\begin{aligned}
\langle \frac{\m_h^{n+1}-\m_h^n}{k} , \Delta_h\frac{\m_h^{n+1}+\m_h^n}{2} \rangle 
= & \frac{h}{2k} \sum_{i=1}^{N_x}(\m_i^{n+1} - \m_i^{n}) \cdot \Delta_h (\m_i^{n+1}+\m_i^n)  \\ 
= &\frac{1}{hk} \sum_{i=1}^{N_x} (\m_i^{n+1} \cdot \m_{i+1}^{n+1} - \m_i^n \cdot \m_{i+1}^n) . 
\end{aligned}
\end{equation}
The last step comes from the fact that 
\[\m_0^{n+1} \cdot \m_1^{n+1} - \m_0^{n}\cdot \m_1^{n} = \m_{N_x}^{n+1} \cdot \m_{N_x+1}^{n+1} - \m_{N_x}^{n}\cdot \m_{N_x+1}^{n} = 0  , \]
due to the discrete Neumann condition \eqref{neumann} and the conservation of magnetization length. From Lemma \ref{dislemma}, we have
\begin{equation}\label{hold3}
	\begin{aligned}
& \langle -\m_h^{n+\frac{1}{2}} \times \Delta_h \m_h^{n+\frac{1}{2}} - \alpha \m_h^{n+\frac{1}{2}}\times (\m_h^{n+\frac{1}{2}}\times \Delta_h \m_h^{n+\frac{1}{2}})
, \Delta_h \m_h^{n+\frac{1}{2}} \rangle \\ 
= & -\alpha \langle  \bm{m}_h^{n+\frac{1}{2}}\times (\bm{m}_h^{n+\frac{1}{2}}\times \Delta_h \bm{m}_h^{n+\frac{1}{2}})
, \Delta_h \m_h^{n+\frac{1}{2}} \rangle   \\
 =&  \alpha  \|\m_h^{n+\frac{1}{2}} \times \Delta_h \m_h^{n+\frac{1}{2}}\|_2^2 .
\end{aligned}
\end{equation}
Therefore, \eqref{enerdif} together with \eqref{hold1}, \eqref{hold2} and \eqref{hold3} implies that 
\[
\begin{aligned}
E_h\left(\m_h^{n+1}\right)-E_h\left(\m_h^{n}\right) 
 =& -2 \alpha k \|\m_h^{n+\frac{1}{2}} \times \Delta_h \m_h^{n+\frac{1}{2}}\|_2^2.
\end{aligned}
\]
In other words, the ICN scheme preserves the energy dissipation law for $\alpha > 0$. The analysis for the \eqref{eq12} is similar.

\noindent \textit{Conditional solvability of the ICN scheme}
\begin{theorem}[{\cite[Theorem 1]{Fuwa2011}}]\label{thm:ICNsolvability}
Assume that $\m_h^n$ satisfies $\| \m_h^n \|_\infty \le 1+\delta$ for some nonnegative $\delta$. We set 
\[\beta < \beta^\ast =\min \left\{ \frac{\eta}{(1+\delta+\eta)^2 + \alpha (1+\delta+\eta)^3|} , \frac{1}{ 2(1+\delta+\eta)+3\alpha(1+\delta+\eta)^2} \right\} \]
for some positive $\eta$. Besides, we introduce the following notation
\[\rho = \beta \cdot (2(1+\delta+\eta) + 3\alpha (1+\delta + \eta)^2) .\] 
If $\frac{k}{h^2} \le \beta$, and $ \m_h^{n,0} $ satisfy 
\[ \| \frac{\m_h^{n,0}+\m_h^n}{2} \|_\infty \le 1 +\delta +\eta .\]
then we have
\[ \| \m_h^{n,\ell+1} -\m_h^{n,\ell} \|_\infty \le \rho^\ell \|\m_h^{n,1} -\m_h^{n,0}\|_\infty ,  \]
\[ \| \frac{\m_h^{n,\ell}+\m_h^n}{2} \|_\infty \le 1 +\delta +\eta .\]
Furthermore, there exists a limit $\m^\ast \in L_h^\infty(\Omega) $ such that
\[ \lim_{\ell\to +\infty} \| \m_h^{n,\ell} -\m^\ast \|_\infty = 0 ,\]
\[ \| \m^\ast \|_\infty \le 1+\delta ,\]
and
\[\frac{\m^\ast- \m_h^n}{k} = - \frac{\m^\ast+ \m_h^n}{2} \times \Delta_h \frac{\m^\ast+ \m_h^n}{2} - \alpha \frac{\m^\ast+ \m_h^n}{2} \times \left( \frac{\m^\ast+ \m_h^n}{2} \times \Delta_h \frac{\m^\ast+ \m_h^n}{2} \right) .\]
\end{theorem}
In this sense, we can define $\m_h^{n+1}$ uniquely by $\m^\ast$. 

\noindent \textit{Convergence of the ICN scheme}
\begin{theorem}[1D case, {\cite[Theorem 5]{Fuwa2011}}]
For the sufficiently smooth solution $ \m_e(x,t) $ of  \cref{LLG,eq7} and the finite difference solution $ \m_i^n $ of \eqref{eq_length}, we define the error on the discrete grids as $ \bm{e}_i^n =\m_i^n - \m_e(ih,nk), i=0, 1,\cdots, N_x, n=0, 1, \cdots, \lfloor \frac{T}{k} \rfloor  $. We assume that $ \m_e(x, 0) \in H^4(\Omega) $ with $\Omega=[0,L]$. If
\[
k \le 1/C_e \left( -1, \alpha, \|\frac{\partial \m_e(x,0)}{\partial x}\|_{H^3(\Omega)} , \|\nabla_h \m_h^0 \|_2^2, 0, L^2,T \right),
\]
and $ C_e $ is a function defined in \cite[(86)]{Fuwa2011}, then we have 
\[\|\e^n\|_{H^1_h} = O(h^2+k^2),\quad \forall\; n \ge 1 .\]
\end{theorem}

\subsection{Semi-implicit Crank-Nicolson scheme}
It is worth mentioning that at each time step, a nonlinear system of equations needs to be solved in \eqref{eq12}. Moreover, at each iteration, a linear system of equations with non-symmetric and variable coefficients has to be solved in Strategy (b) and Strategy (c). To overcome this subtle difficulty, we approximate the nonlinear terms in front of the Laplacian operator by using available data from previous time steps (one-sided interpolation), with the same accuracy order as in the ICN scheme. Precisely, for \eqref{eq12}, we have
\begin{align*}
\frac{{\m}_h^{n+1}-{\m}_h^{n}}{k}
=& -\hat{\m}_h^{n+\frac{1}{2}}\times\big(\epsilon \Delta_h\m_h^{n+\frac{1}{2}} +\hat{\f}_h^{n+\frac{1}{2}} \big) \nonumber\\
&  -\alpha\hat{\m}_h^{n+\frac{1}{2}}\times\left(\hat{\m}_h^{n+\frac{1}{2}}\times(\epsilon \Delta_h\m_h^{n+\frac{1}{2}} +\hat{\f}_h^{n+\frac{1}{2}} ) \right),
\end{align*}
where
\begin{align*}
    \hat{\m}_h^{n+\frac12} = \frac{3{\m}_h^{n}-{\m}_h^{n-1}}{2} \label{m_hat}\qquad\textrm{and}\qquad 
    \hat{\f}_h^{n+\frac12} = \frac{3{\f}_h^{n}-{\f}_h^{n-1}}{2}.
\end{align*}
Such a scheme does not preserve the length of magnetization. Therefore, we add a projection step and obtain the following SICN scheme for \eqref{eq10}: 
\begin{equation}\label {SICNs}
\left\{ 
\begin{aligned}
\frac{{\m}_h^{n+1,*}-{\m}_h^{n}}{k} & =  -\hat{\m}_h^{n+\frac{1}{2}}\times\big(\epsilon \Delta_h\m_h^{n+\frac{1}{2},*} +\hat{\f}_h^{n+\frac{1}{2}} \big) \\
&\quad -\alpha\hat{\m}_h^{n+\frac{1}{2}}\times\left(\hat{\m}_h^{n+\frac{1}{2}}\times(\epsilon \Delta_h\m_h^{n+\frac{1}{2},*} +\hat{\f}_h^{n+\frac{1}{2}} ) \right), \\
\m_h^{n+1} & = \frac{{\m}_h^{n+1,*}}{ |{\m}_h^{n+1,*}| },
\end{aligned}
\right.
\end{equation} 
where ${\m}_h^{n+1,*}$ is the intermediate magnetization and $ \displaystyle \m_h^{n+\frac{1}{2},*}  = \frac{{\m}_h^{n+1,*}+{\m}_h^{n}}{2}$ .

\noindent \textit{Unconditionally unique solvability of the SICN scheme}

For ease of notation, we drop the temporal indices and rewrite the linear system of \eqref{SICNs} in a compact form as
\begin{equation}\label{eq_comp}
	(2I + k\epsilon \hat{\bm{m}}_h \times \Delta_h + \alpha k \epsilon \hat{\bm{m}}_h \times (\hat{\bm{m}}_h \times \Delta_h) ) \m_h = \p_h,
\end{equation}
where $\p_h$, and $\hat{\bm{m}}_h$ are given.

\begin{theorem}\label{thm:SICNsolvability}
	Given $\p_h$ and $\hat{\m}_h $, the numerical scheme \eqref{eq_comp} admits a unique solution.
\end{theorem}

\begin{proof}
For the unique solvability analysis for \eqref{eq_comp}, we denote $\q_h = -\Delta_h \m_h$. Note that $\overline{\q_h}=0$ under the Neumann boundary condition for $\m_h$. In general, $\m_h \neq (-\Delta_h)^{-1} \q_h$ since $\overline{\m_h} \neq 0$. Instead, we reformulate \eqref{eq_comp} as
\[\m_h = \frac{1}{2} \left(\p_h + k\epsilon \hat{\m}_h \times \q_h + \alpha k \epsilon \hat{\m}_h \times (\hat{\m}_h \times \q_h) \right)\] , 
and take an average on both sides. Therefore, $\m_h$ can be represented as follows:
\[\m_h = (-\Delta_h)^{-1} \q_h + C^{\ast}_{\q_h}  \ \textrm{with} \ C^{\ast}_{\q_h} = \frac{1}{2} \left(\overline{\p_h} + k\epsilon \overline{\hat{\m}_h \times \q_h}+ \alpha k \epsilon \overline{\hat{\m}_h \times(\hat{\m}_h \times \q_h)}  \right).\] 
In turn, \eqref{eq_comp} is rewritten as
\[G(\q_h) := 2 \left( (-\Delta_h)^{-1} \q_h + C_{\q_h}^{\ast} \right) -\p_h -k \epsilon \hat{\m}_h \times \q_h - \alpha k \epsilon \hat{\m}_h \times (\hat{\m}_h \times \q_h) = \bm{0}. \]

For any $\q_{1,h}, \q_{2,h}$ with $\overline{\q_{1,h}} = \overline{\q_{2,h}} = 0$, we denote $\tilde{\q}_h = \q_{1,h} - \q_{2,h}$ and derive the following monotonicity estimate:
\begin{align*}
\left\langle G\left(\bm{q}_{1, h}\right)-G\left(\bm{q}_{2, h}\right), \bm{q}_{1, h}-\bm{q}_{2, h}\right\rangle 
=&2\left(\left\langle\left(-\Delta_{h}\right)^{-1} \tilde{\bm{q}}_{h}, \tilde{\bm{q}}_{h}\right\rangle+\left\langle C_{\bm{q}_{1, h}}^{*}-C_{\bm{q}_{2, h}}^{*}, \tilde{\bm{q}}_{h}\right\rangle\right) \\
&-k \epsilon\left\langle\hat{\bm{m}}_{h} \times \tilde{\bm{q}}_{h}, \tilde{\bm{q}}_{h}\right\rangle-\alpha k \epsilon\left\langle\hat{\bm{m}}_{h} \times\left(\hat{\bm{m}}_{h} \times \tilde{\bm{q}}_{h}\right), \tilde{\bm{q}}_{h}\right\rangle \\
\geq& 2\left(\left\langle\left(-\Delta_{h}\right)^{-1} \tilde{\bm{q}}_{h}, \tilde{\bm{q}}_{h}\right\rangle+\left\langle C_{\bm{q}_{1, h}}^{*}-C_{\bm{q}_{2, h}}^{*}, \tilde{\bm{q}}_{h}\right\rangle\right) \\
=&2\left\langle\left(-\Delta_{h}\right)^{-1} \tilde{\bm{q}}_{h}, \tilde{\bm{q}}_{h}\right\rangle=2\left\|\tilde{\bm{q}}_{h}\right\|_{-1}^{2} \geq 0 .
\end{align*}
Note that the following equality and inequality have been used in the second step:
\[
\left\langle\hat{\bm{m}}_{h} \times \tilde{\bm{q}}_{h}, \tilde{\bm{q}}_{h}\right\rangle=0, \quad\left\langle\hat{\bm{m}}_{h} \times\left(\hat{\bm{m}}_{h} \times \tilde{\bm{q}}_{h}\right), \tilde{\bm{q}}_{h}\right\rangle \leq 0.
\]
The third step is based on the fact that $\left\langle C_{\bm{q}_{1, h}}^{*}-C_{\bm{q}_{2, h}}^{*}, \tilde{\bm{q}}_{h}\right\rangle=0$ since both $C_{\bm{q}_{1, h}}^{*}$ and $C_{\bm{q}_{2, h}}^{*}$ are constants and $\overline{\bm{q}_{1, h}}=\overline{\bm{q}_{2, h}}=0$.
Furthermore, for any $\bm{q}_{1, h}, \bm{q}_{2, h}$ with $\overline{\bm{q}_{1, h}}=\overline{\bm{q}_{2, h}}=0,$ we have
\[
\left\langle G\left(\bm{q}_{1, h}\right)-G\left(\bm{q}_{2, h}\right), \bm{q}_{1, h}-\bm{q}_{2, h}\right\rangle \geq 2\left\|\tilde{\bm{q}}_{h}\right\|_{-1}^{2}>0 \quad \text { if } \bm{q}_{1, h} \neq \bm{q}_{2, h},
\]
and the equality only holds when $\bm{q}_{1, h}=\bm{q}_{2, h}$. Finally, an application of the Browder-Minty lemma \cite{Browder1965,Minty1963} implies a unique solution of the SICN scheme. The proof here mainly follows the same way as Theorem 2.1 in \cite{Chen2021}.
\end{proof}

\noindent\textit{Convergence analysis of the SICN scheme}

\begin{theorem}\label{thm:convergenceSICN}
Let $\m_e(\x, t)\in C^3 ([0,T]; C^0) \cap L^{\infty}([0,T]; C^4)$ be a smooth solution of \cref{LLG,eq7} with the initial data $\m_e (\x,0)$, and $\m_h$ be the numerical solution of the equation \eqref{eq_length} with the initial data ${\m}_h^0=\m_{e}(\x_h, 0)$ for the numerical grid $\x_h$. Define the error as $ \bm{e}^n =\m_h^n - \m_e(\x_h,t) $. Suppose that the initial error satisfies $\|\bm{e}^0 \|_{H_h^1} + \|\bm{e}^1 \|_{H_h^1}= \mathcal{O} (k^2 + h^2)$, and $k\leq \mathcal{C}h$. Then the following convergence result holds true as $h$ and $k$ go to zero:
	\begin{align*}
	\| \bm{e}^n  \|_{H^1_h} &\leq \mathcal{C}(k^2+h^2), \quad \forall\; n \ge 2 ,
	\end{align*}	
	in which the constant $\mathcal{C}>0$ is independent of $k$ and $h$.
\end{theorem}
The proof of Theorem \ref{thm:convergenceSICN} basically follows the proof of \cite[Theorem 2.3]{Chen2021} and a precise estimate between $\bm{m}^n_h$ and the intermediate magnetization $\bm{m}^{n,*}_h$. The idea is pretty much the same while the entire procedure is tedious, so we omit the details here. Note that the convergence analysis for both the ICN and SICN schemes neglects the stray field in the effective field for technical simplification.

\section{Numerical examples}\label{sec:result}
In this section, we perform some numerical tests to investigate the accuracy and efficiency for the model problem of the LLG equation. For simplicity, we only consider the exchange field with $\epsilon=1$ in \eqref{heff} for the 1D and 3D cases. When solving nonlinear systems of equations in the ICN scheme, the maximum number of iterations in the Newton's method is set to be $MaxIter = 300$ with the tolerance to stop the iteration being $tol = 1e-12$. The error is recorded as the difference between the numerical and exact solutions $\left\|\bm{m}_h -\bm{m}_e\right\|_{\infty}$. Results in the discrete $L^2$ and $H^1$ norms are similar, so we do not include all the details here. For completeness, the number of iterations in the Newton's method per time step is recorded for the ICN scheme as well. 

In the 1D model, we choose $\bm{m}_{e}=(\cos (\bar{x}) \sin (t), \sin (\bar{x}) \sin (t), \cos (t))$ with $\bar{x}=x^{2}(1-x)^{2}$ as the exact solution over $ \Omega = [0,1] $, which satisfies
\begin{equation}\label{eqn:exact1D}
\bm{m}_{t}=-\bm{m} \times \bm{m}_{xx}-\alpha \bm{m} \times(\bm{m} \times \bm{m}_{xx})+\bm{g}.
\end{equation}
The forcing term turns out to be $\mathbf{g}=\bm{m}_{e t}+\bm{m}_{e} \times \bm{m}_{e x x}+\alpha \bm{m}_{e} \times\left(\bm{m}_{e} \times \bm{m}_{e x x}\right)$, and $\bm{m}_{e}$ satisfies the homogeneous Neumann boundary condition.

In the 3D model, we set the exact solution as 
\[
\bm{m}_{e}=(\cos (\bar{x} \bar{y} \bar{z}) \sin (t), \sin (\bar{x} \bar{y} \bar{z}) \sin (t), \cos (t)),
\]
over $\Omega=[0,1]^3$, which satisfies the homogeneous Neumann boundary condition and the following equation is valid: 
\begin{equation}\label{eqn:exact3D}
\bm{m}_{t}=-\bm{m} \times \Delta \bm{m}-\alpha \bm{m} \times(\bm{m} \times \Delta \bm{m})+\bm{g} , 
\end{equation}
with $\bar{x}=x^{2}(1-x)^{2}, \bar{y}=y^{2}(1-y)^{2}, \bar{z}=z^{2}(1-z)^{2}$ and $\bm{g}=\bm{m}_{e t}+\bm{m}_{e} \times \Delta \bm{m}_{e}+\alpha \bm{m}_{e} \times\left(\bm{m}_{e} \times \Delta \bm{m}_{e}\right)$ .

\subsection{Accuracy tests}
In the 1D computation, we fix $h = 1/2400$ and record the approximation errors in terms of the temporal step size $ k $ in \cref{tab:2}. It is clear that both the ICN and SICN schemes are second-order accurate in time. To get the spatial accuracy, we fix $ k  = 5e-7$ and record the approximation errors in terms of the spatial mesh size $h$ . It follows from \cref{tab:3} that both schemes are second-order accurate in space. The conclusions could be obtained the the 3D comoutation; see \cref{tab:6} with $ h = 0.025$ and \cref{tab:7} with $ k  = 1e-3 $.
\begin{table}[htp]
	\caption{Numerical errors of ICN and SICN schemes in terms of $ k $, with $ T = 1 $, $ h =1/2400 $, and $ \alpha = 0.00001 $ in the 1D computation.}
	\label{tab:2}
	\centering
	\begin{threeparttable}
		\begin{tabular}{l|cccccc}
			\Xhline{1.5pt}
			$ k $ & T/120 & T/130 & T/140 &T/150 & order \\
			\hline
			ICN scheme &3.3336e-06 &2.8530e-06 &  2.4718e-06 &2.1642e-06 & 1.9361\\
			\hline
			SICN scheme &2.9816e-06 &2.5549e-06&  2.2151e-06&1.9407e-06&1.9246 \\
			\hline
			Iterations\tnote{*} & 23 & 12 & 9 & 7 & / \\
			\Xhline{1.5pt}
		\end{tabular}
		\begin{tablenotes}
		\footnotesize
		\item[*] Iterations refer to the maximum number of iterations in the Newton's method over all temporal steps. The number of iterations remains almost a constant from step to step since the same initialization strategy is used. Linear systems at each iteration in the ICN method and the linear system in the SICN method are solved by the direct method. The same notation is used in \crefrange{tab:3}{tab:11}.
		\end{tablenotes}

	\end{threeparttable}
\end{table}
\begin{table}[htp]
	\caption{Numerical errors of ICN and SICN schemes in terms of $ h $, with $ T = 5e-2 $, $ k =5e-7 $, and $ \alpha = 0.00001 $ in the 1D computation.}
	\label{tab:3}
	\centering
	\begin{threeparttable}
		\begin{tabular}{l|cccccc}
			\Xhline{1.5pt}
			$ h $ & 1/50 & 1/60 & 1/70 &1/80 & order \\
			\hline
			ICN scheme &1.2036e-06 &8.3629e-07 &  6.1427e-07 &4.7025e-07 & 1.9997 \\
			\hline
			SICN scheme &1.2033e-06 & 8.3592e-07&  6.1387e-07&4.6984e-07 & 2.0010 \\
			\hline
			Iterations & 2 & 2 & 2 & 2 &/ \\
			\Xhline{1.5pt}
		\end{tabular}
	\end{threeparttable}
\end{table}

\begin{table}[htp]
	\centering
	\caption{Numerical errors of ICN and SICN schemes in terms of $ k $, with $ T = 1 $, $ h =0.025 $, and $ \alpha = 0.00001 $ in the 3D computation.}
	\label{tab:6}
	\begin{threeparttable}
		\begin{tabular}{l|ccccc}
			\Xhline{1.5pt}
			$ k $ & T/6 & T/8 & T/10 & T/12 & order  \\
			\hline
			ICN scheme &1.1133e-03  & 6.2644e-04 &  4.0121e-04 &  2.7886e-04 &1.9973 \\
			\hline
			SICN scheme &0.9353e-03 &5.5618e-04 &  3.6801e-04&2.6136e-04 & 1.8389\\
			\hline
			Iterations & 4 & 4 & 4 & 4 & / \\ 
			\Xhline{1.5pt}
		\end{tabular}
	\end{threeparttable}
\end{table}
\begin{table}[htp]
	\centering
	\caption{Numerical errors of ICN and SICN schemes in terms of $ h $, with $ T = 0.1 $, $ k =1e-3 $, and $ \alpha = 0.00001 $ in the 3D computation.}
	\label{tab:7}
	\begin{threeparttable}
		\begin{tabular}{l|cccccc}
			\Xhline{1.5pt}
			$ h $ & 1/10 & 1/12 & 1/14 &1/16 & order \\
			\hline
			ICN scheme &5.8310e-07 &4.1038e-07  &   3.0434e-07  &2.3454e-07 & 1.9375 \\
			\hline
			SICN scheme &5.8323e-07 & 4.1048e-07&  3.0443e-07&2.3462e-07 &1.9372 \\
			\hline
			Iterations & 3 & 3 & 3 & 3 & / \\ 
			\Xhline{1.5pt}
		\end{tabular}
	\end{threeparttable}
\end{table}

From \cref{tab:2} to \cref{tab:7}, we find that the semi-implicit nature of of the ICN scheme reduces the number of iterations from $2$ or more to $1$, at each time step, with almost the same numerical error. As a consequence, the SICN scheme reduces the CPU time by at least $50\%$, in comparison with the ICN scheme for the same accuracy requirement, as will be demonstrated in the next subsection.

\subsection{Efficiency comparison}
In \cref{fig:1d}, we plot the CPU time (in seconds) of the ICN and SICN schemes with respect to the numerical error $ \|\m_h - \m_e\|_\infty $ in the 1D and 3D computations, respectively. These results are consistent with the ones presented in \crefrange{tab:2}{tab:7}, which have demonstrated that the SICN scheme saves at least $50\%$ CPU time without an accuracy sacrifice. 
\begin{figure}[htbp]
	\centering
	\begin{minipage}[t]{0.49\textwidth}
		\centering
		\includegraphics[scale=0.4]{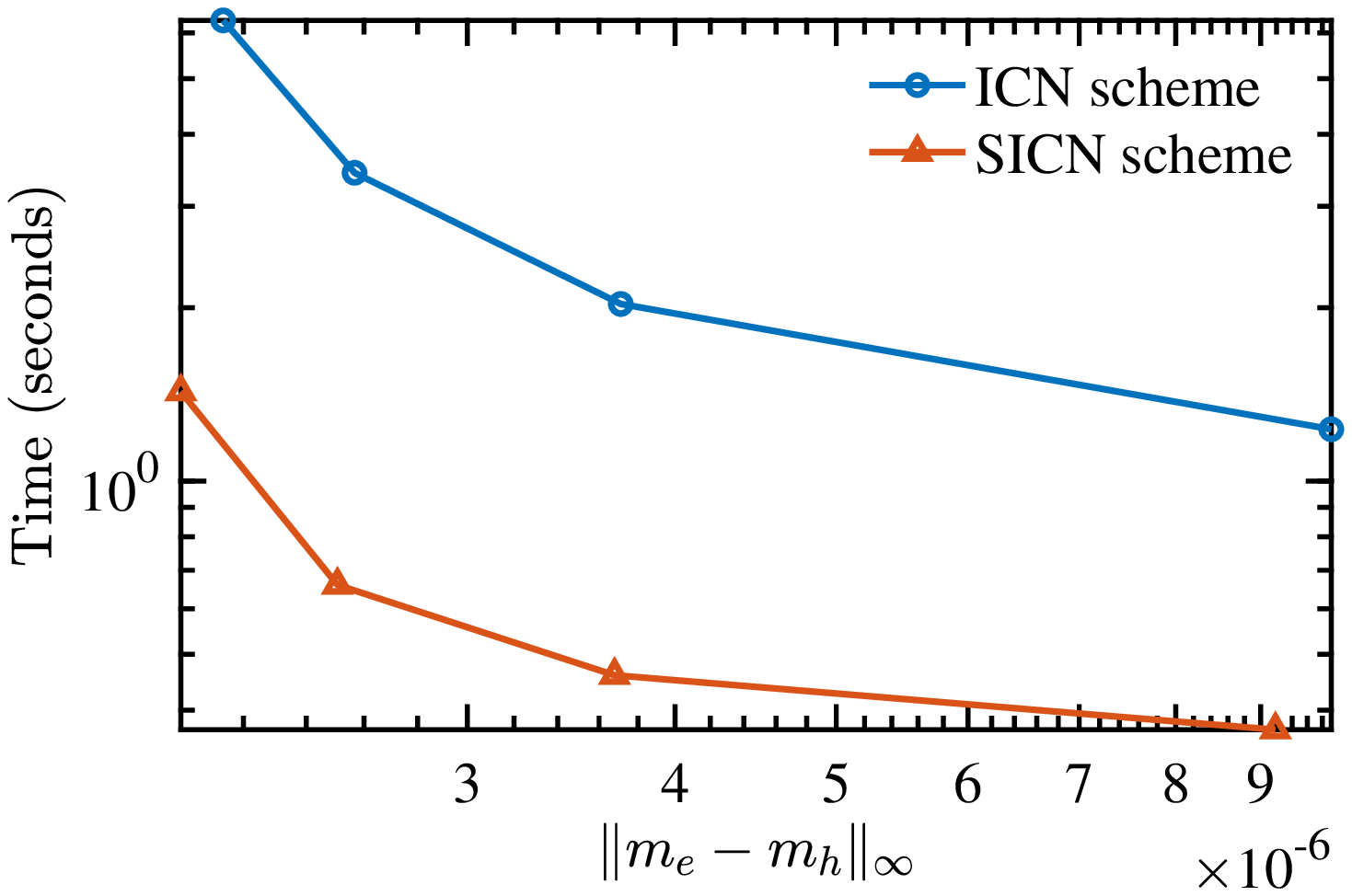}
		\centerline{(a) $ k $}
	\end{minipage}
	\begin{minipage}[t]{0.49\textwidth}
		\centering
		\includegraphics[scale=0.4]{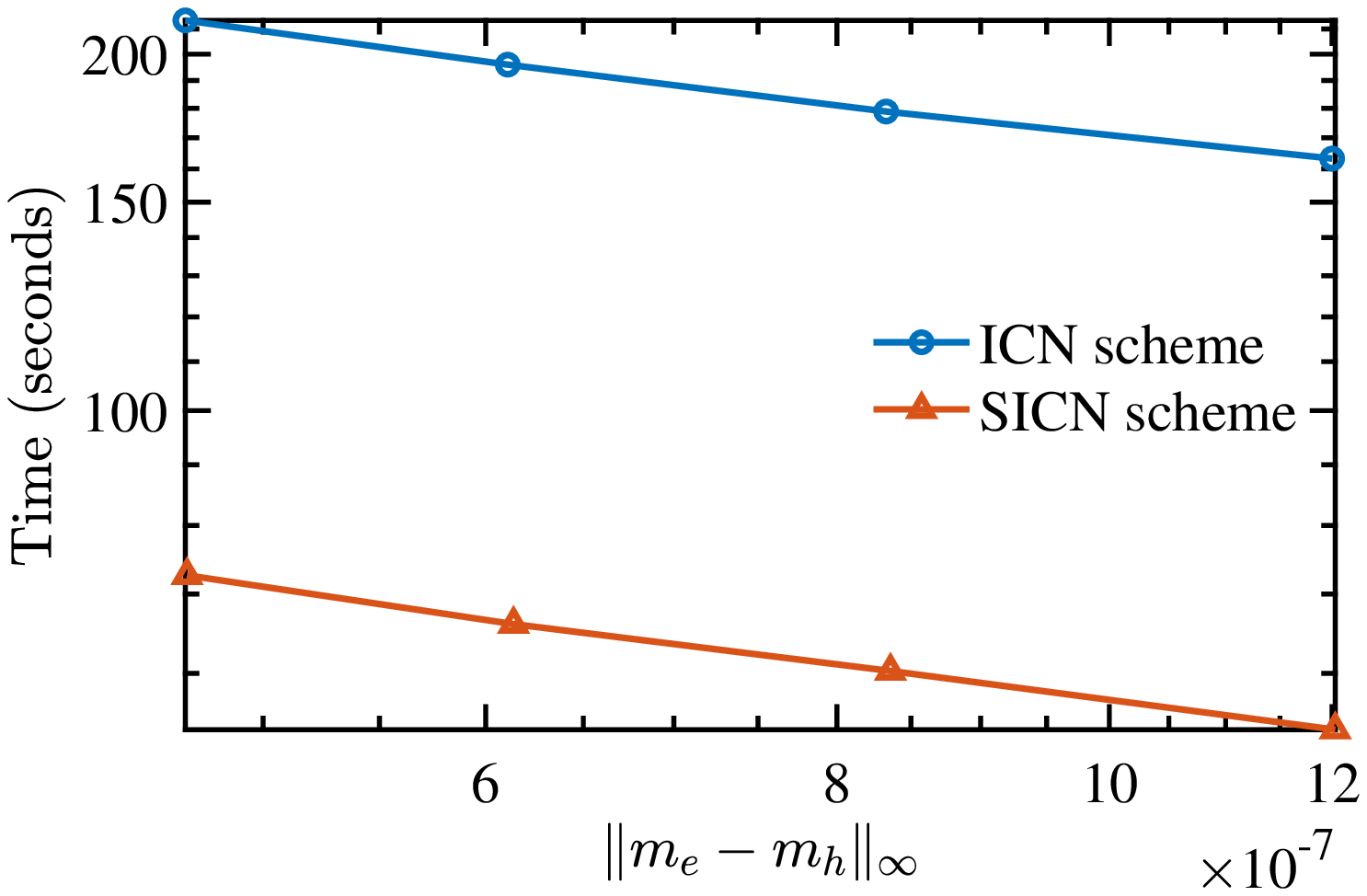}
		\centerline{(b) $ h $ }
	\end{minipage}
	\centering
\begin{minipage}[t]{0.49\textwidth}
	\centering
	\includegraphics[scale=0.4]{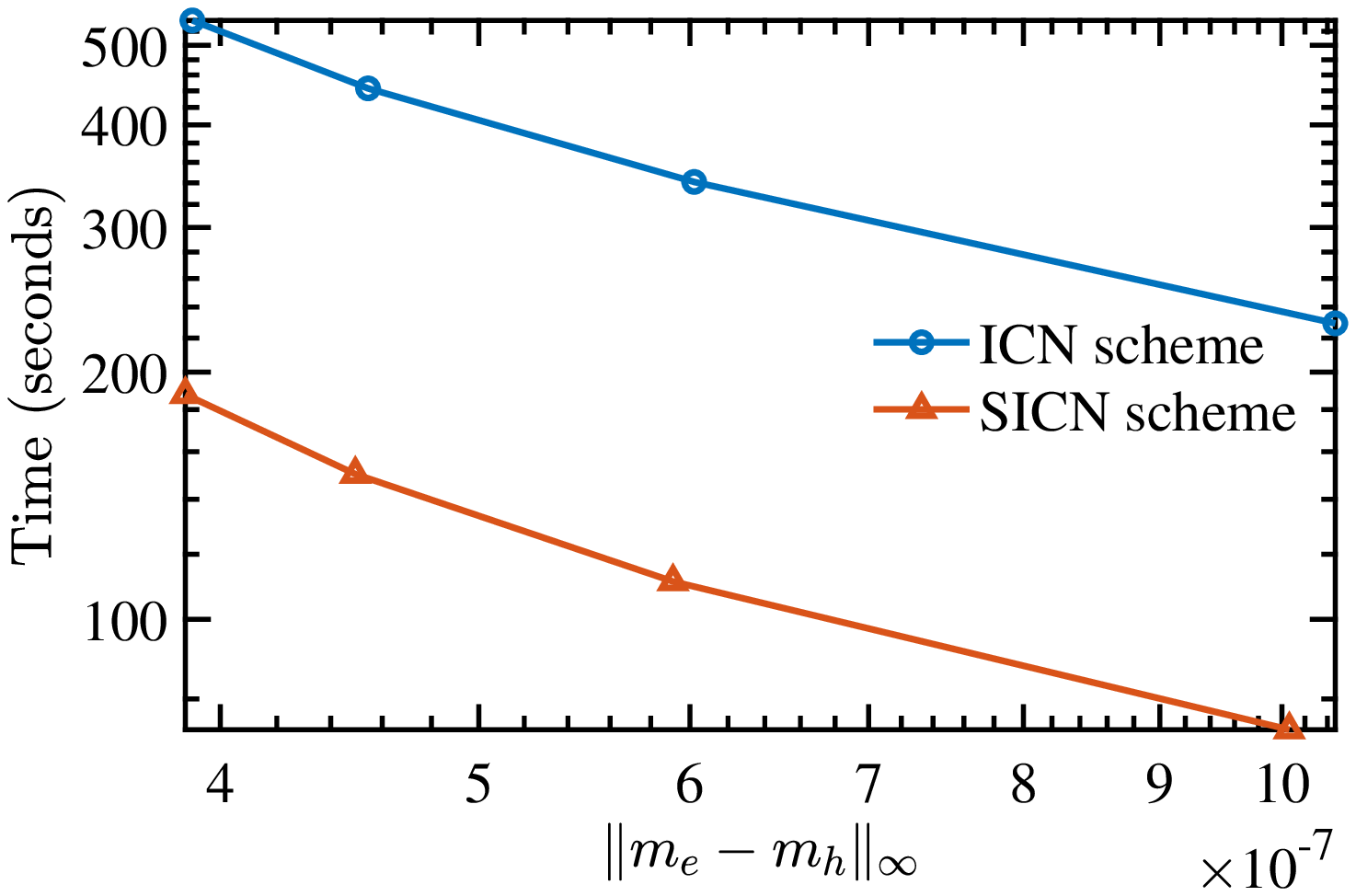}
	\centerline{(c) $ k $}
\end{minipage}
\begin{minipage}[t]{0.49\textwidth}
	\centering
	\includegraphics[scale=0.4]{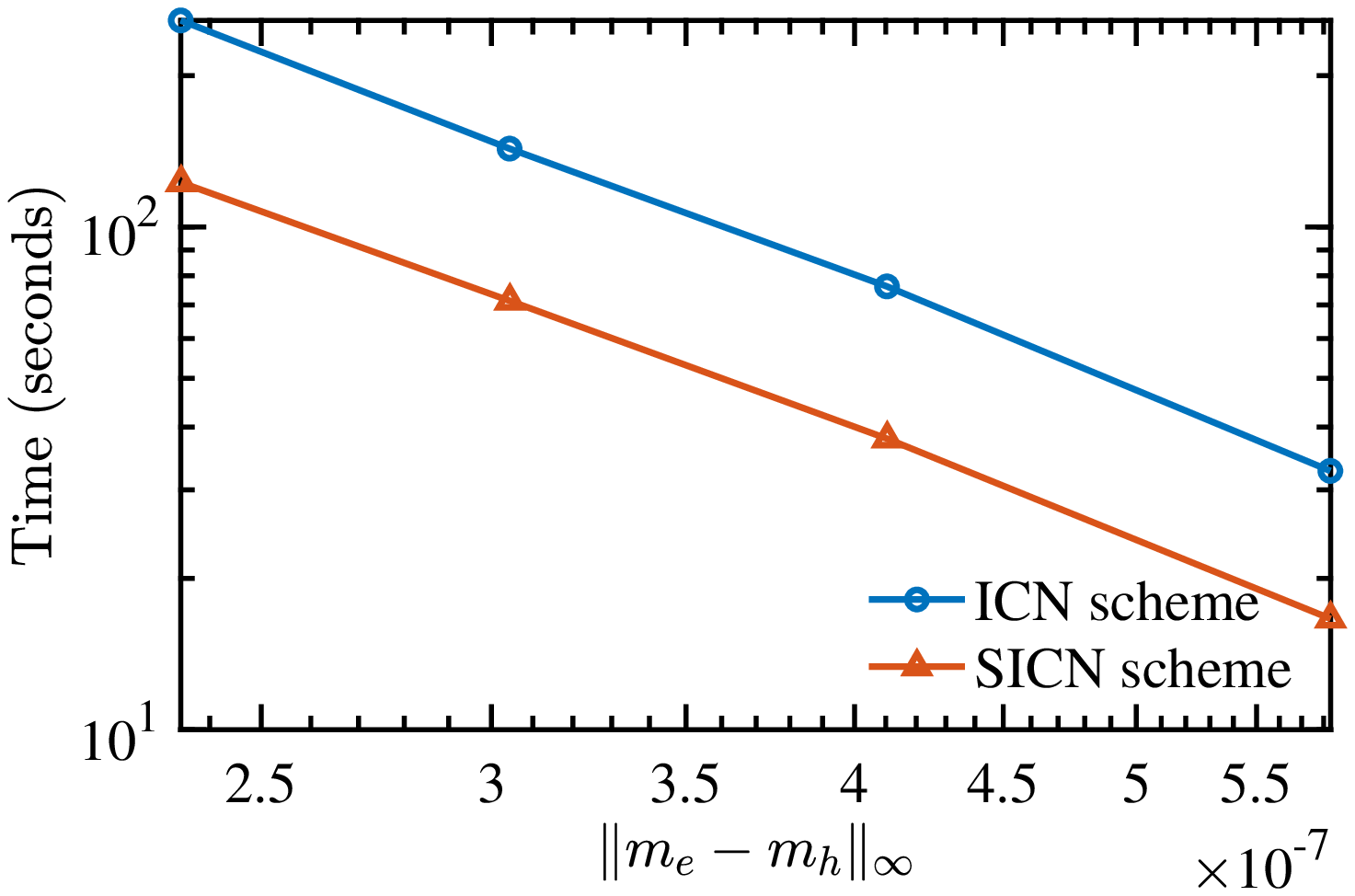}
	\centerline{(d) $ h $ }
\end{minipage}
	\caption{CPU time (in seconds) with respect to the numerical error by varying $k$ and $h$, respectively. Top row: the 1D case ; Bottom row: the 3D case.}
	\label{fig:1d}
\end{figure}

\subsection{Limitations of Newton's method in the ICN scheme}
The usage of gradient information allows Newton's method to converge quickly, while the convergence depends on the initial guess, mesh size, and time step size. 
Generally speaking, a larger time step size (relative to the mesh size) makes Newton's method more difficult to converge. 

For the 1D computation, in which the exact profile satisfies \eqref{eqn:exact1D}, we test the convergence of Newton's method with different setups. Results are recorded in \cref{tab:10}. In Cases 1 and 2, Newton's method does not converge with $\lambda=1$ but converges to the exact solution with a smaller $\lambda=0.1$. In Case 3, 
Newton's method converges with the initial guess $\m_h^n$. To test the convergence domain of Newton's method, a random unit vector field (uniform distribution on $[-1,1]$ with a projection onto $|\m|=1$) is used as the initial guess in Case 4. Approximately 25\% of the initial guesses have achieved the convergence.
\begin{table}[htp]
\centering
\caption{Convergence of Newton's method in the 1D model, with $T=1, tol=1e-12$, and $MaxIter=300$.}
\label{tab:10}
\begin{tabular}{c|c|c|c|c}
\Xhline{1.5pt}
Case   & 1      & 2      & 3  & 4   \\ \hline
$ (h,k) $ &  $(1/800,1/10)$ & $(1/800,1/10)$ & $(1/60,1e-3)$ & $(1/60,1e-3)$ \\ \hline
Initial guess & $\m_h^n$ & $\m_h^n$ & $\m_h^n$ & random \\ \hline
$\lambda$     & 1      & 0.1    & 1   & 1     \\ \hline
Convergence   & No     & Yes    & Yes & 25\% (by chance)    \\ \hline
Iterations & /      &  242 &   3 & $\sim 9$   \\ \hline
Numerical error         &  /      &  4.6992e-04      &   1.3559e-04  & /      \\ \Xhline{1.5pt}
\end{tabular}
\end{table}

For the 3D computation, in which the exact profile satisfies \eqref{eqn:exact3D}, we test the convergence of Newton's method with different setups. Results are recorded in \cref{tab:11}. In Case 5, Newton's method converges with the initial guess $\m_h^n$. To test the convergence domain of Newton's method, a random unit vector field (uniform distribution on $[-1,1]$ with a projection onto $|\m|=1$) is used as the initial guess in Case 6. Approximately 21\% of the initial guesses have achieved the convergence.
\begin{table}[htp]
\centering
\begin{threeparttable}
\caption{Convergence of Newton's method in the 3D model, with $T=1, tol=1e-12$, and $MaxIter=30$.}
\label{tab:11}
\begin{tabular}{c|c|c}
\Xhline{1.5pt}
Case     & 5  & 6   \\ \hline
$ (h,k) $ & $(1/10,1/55)$ & $(1/10,1/55)$ \\ \hline
Initial guess & $\m_h^n$ & random \\ \hline
$\lambda$     & 1   & 1     \\ \hline
Convergence   & Yes & 21\% (by chance)    \\ \hline
Iterations 	  &   3 & $\sim 10$   \\ \hline
Error         &   2.8369e-05  & /      \\ \Xhline{1.5pt}
\end{tabular}
\end{threeparttable}
\end{table}

The above results are obtained for large temporal step sizes, i.e., $k=1/10$. In the practical  applications, we usually use smaller step sizes to balance the temporal error and the spatial error. Therefore, the non-uniqueness of solutions in an implicit scheme or the divergence of Newton's method can be avoided; see \crefrange{tab:2}{tab:7}.

\section{Micromagnetics simulations}\label{sec:micromagnetics}
Micromagnetics simulations require the evaluation of the stray field \eqref{eqq-5}, which is often implemented by the fast Fourier transform \cite{Wang2001}. In an implicit scheme, several evaluations of the stray field are needed per time step. Therefore, it is inefficient to use the ICN scheme in this scenario. In this section, we use the SICN scheme to conduct micromagnetics simulations, including different stable structures and a benchmark problem from NIST. The non-symmetric linear systems of equations in the SICN scheme for full LLG equation are solved by GMRES from hypre \cite{Falgout2002}. 

\subsection{Equilibrium states}
\label{sec:states}

We apply the SICN scheme, with a spatial resolution $ 64\times 128 \times 1$, on a $ 1 \times 2  \times 0.02 \ \mathrm{\mu m}^3 $ thin-film element, with material parameters of Permalloy in \cref{tab1} and the damping parameter $\alpha=0.1$. A fine temporal step size $1$ picosecond (ps) is used to observe the evolution of magnetization distributions. In the absence of an external field, multiple meta-stable states are often observed in ferromagnets, both experimentally and numerically \cite{Nakatani1989, Zheng1997, Schrefl2003, Hertel1998}. Using the SICN scheme, we obtain four equilibrium states with four different initial magnetization distributions. They are Diamond state, Landau state, C-state, and S-state; see \cref{fig:state}.
\begin{figure}[htp]
\centering
\begin{minipage}[t]{0.49\textwidth}
\centering
\includegraphics[width=6cm]{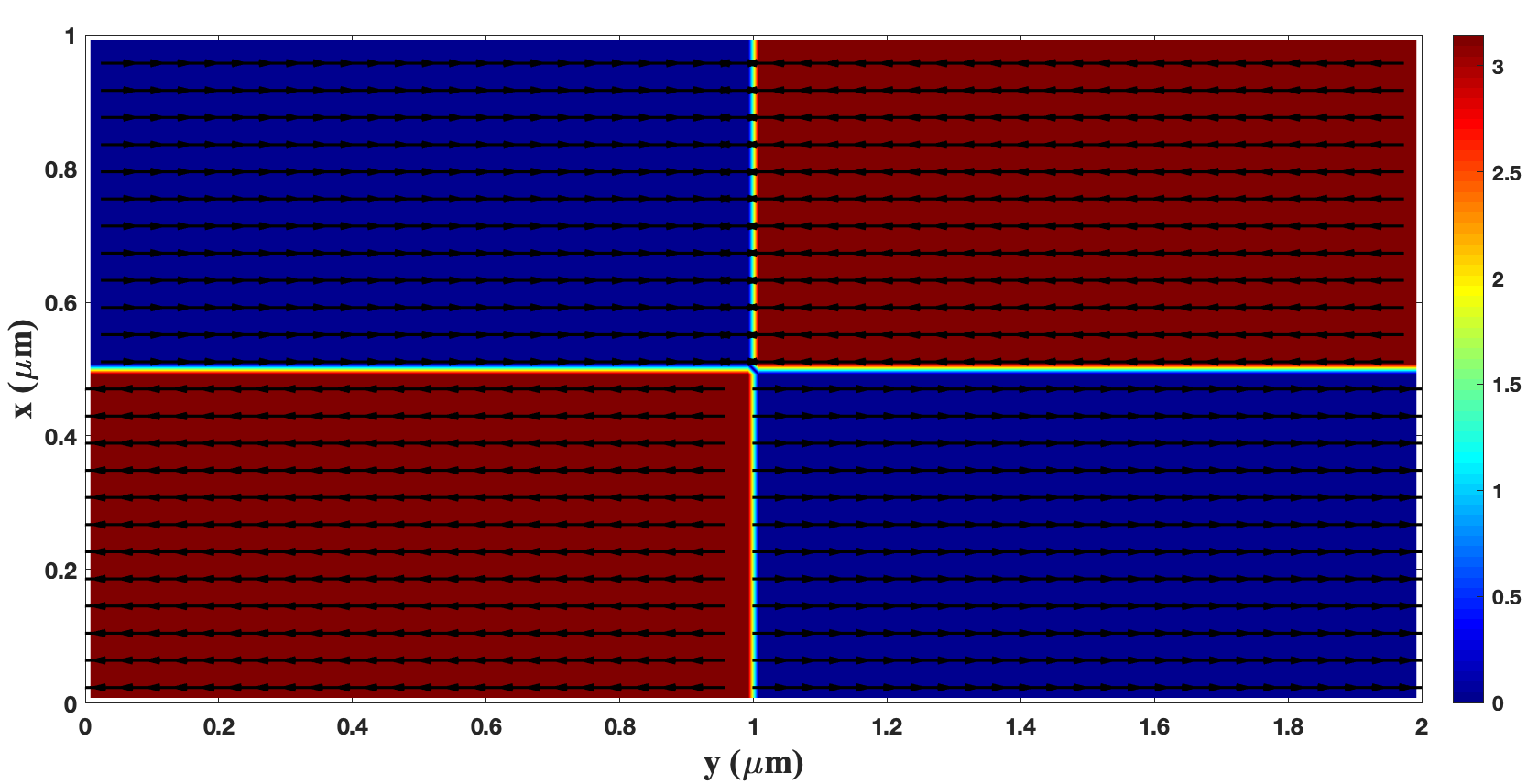}
\end{minipage}
\begin{minipage}[t]{0.49\textwidth}
\centering
\includegraphics[width=6cm]{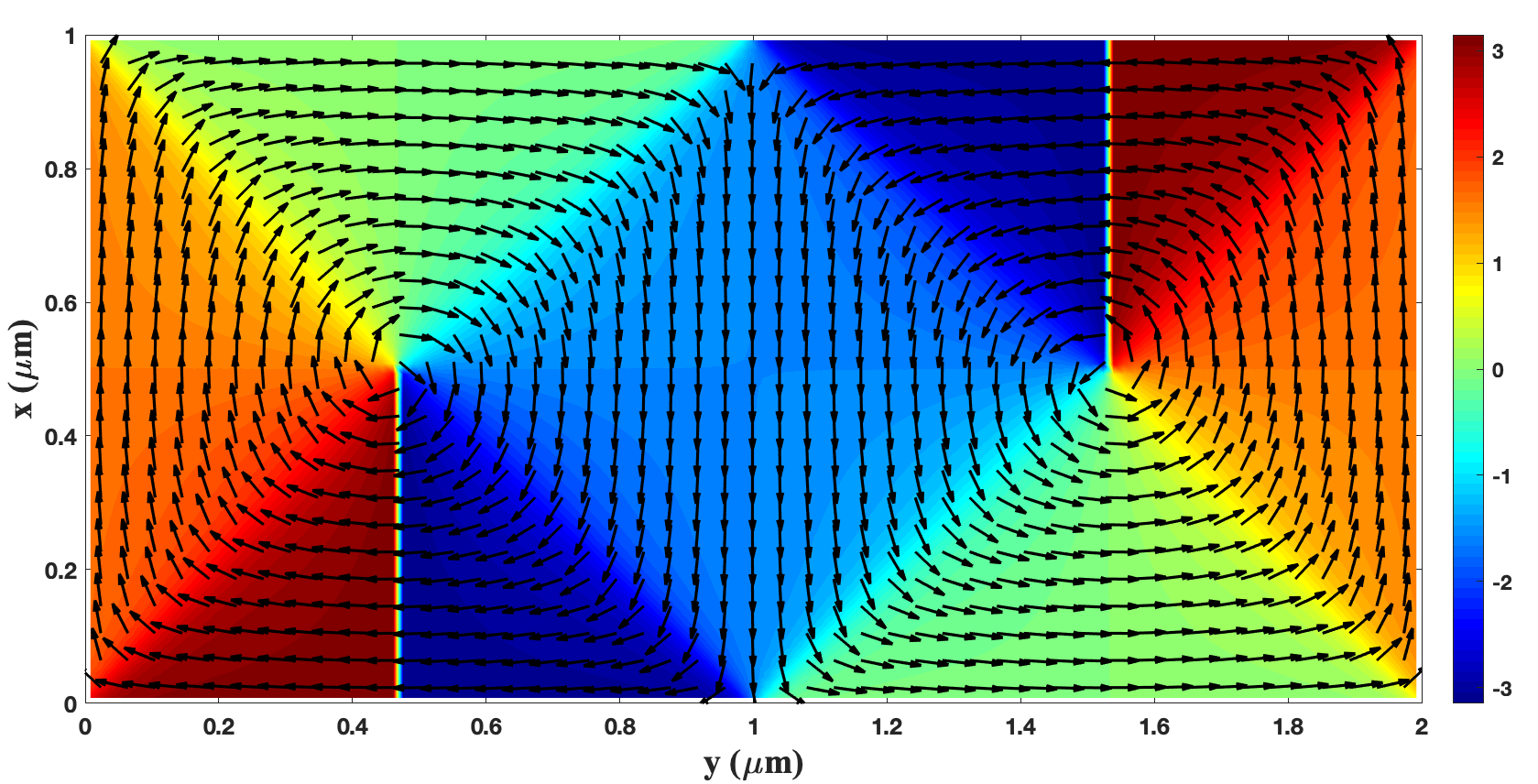}
\end{minipage}
\centering
\begin{minipage}[t]{0.49\textwidth}
\centering
\includegraphics[width=6cm]{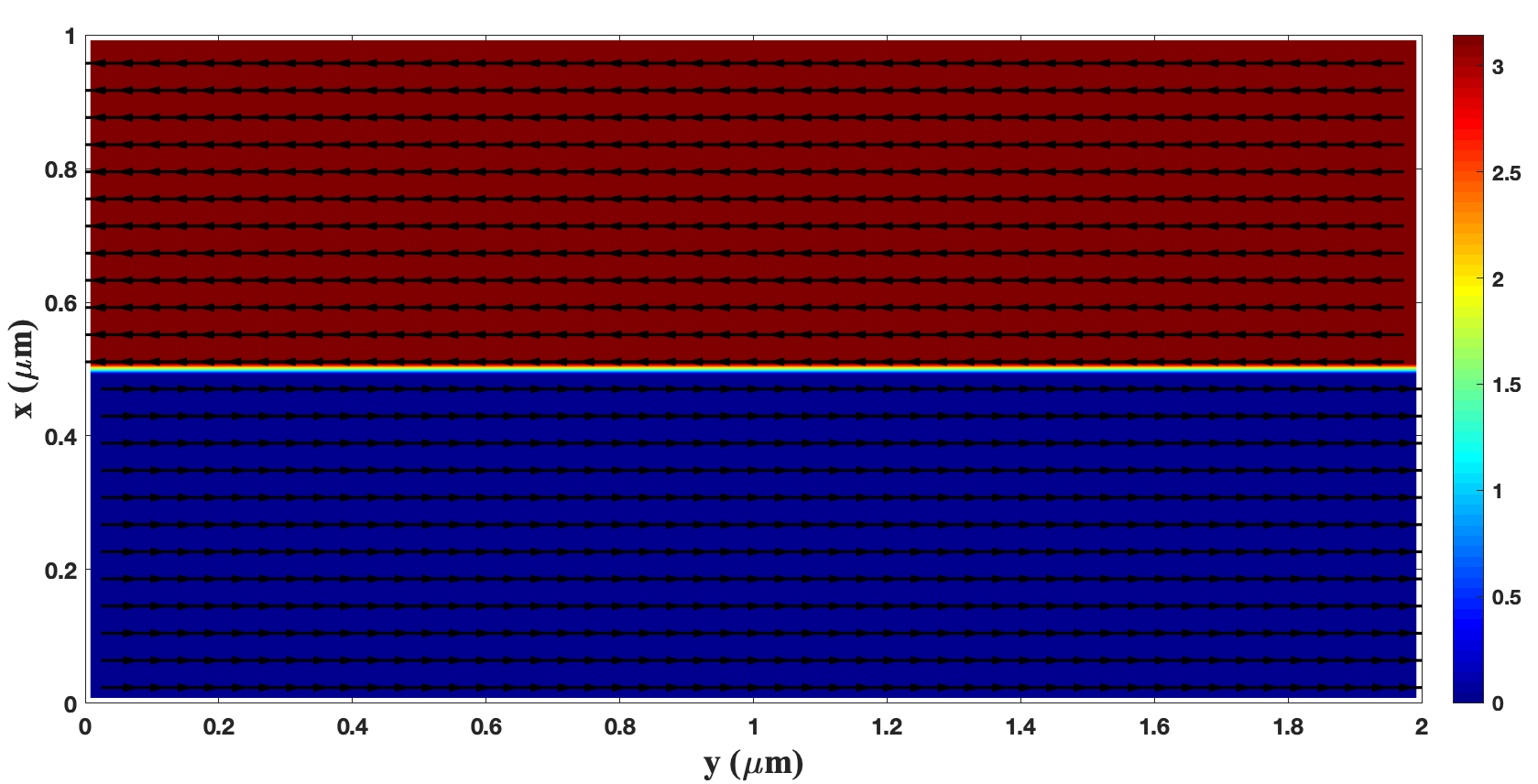}
\end{minipage}
\begin{minipage}[t]{0.49\textwidth}
\centering
\includegraphics[width=6cm]{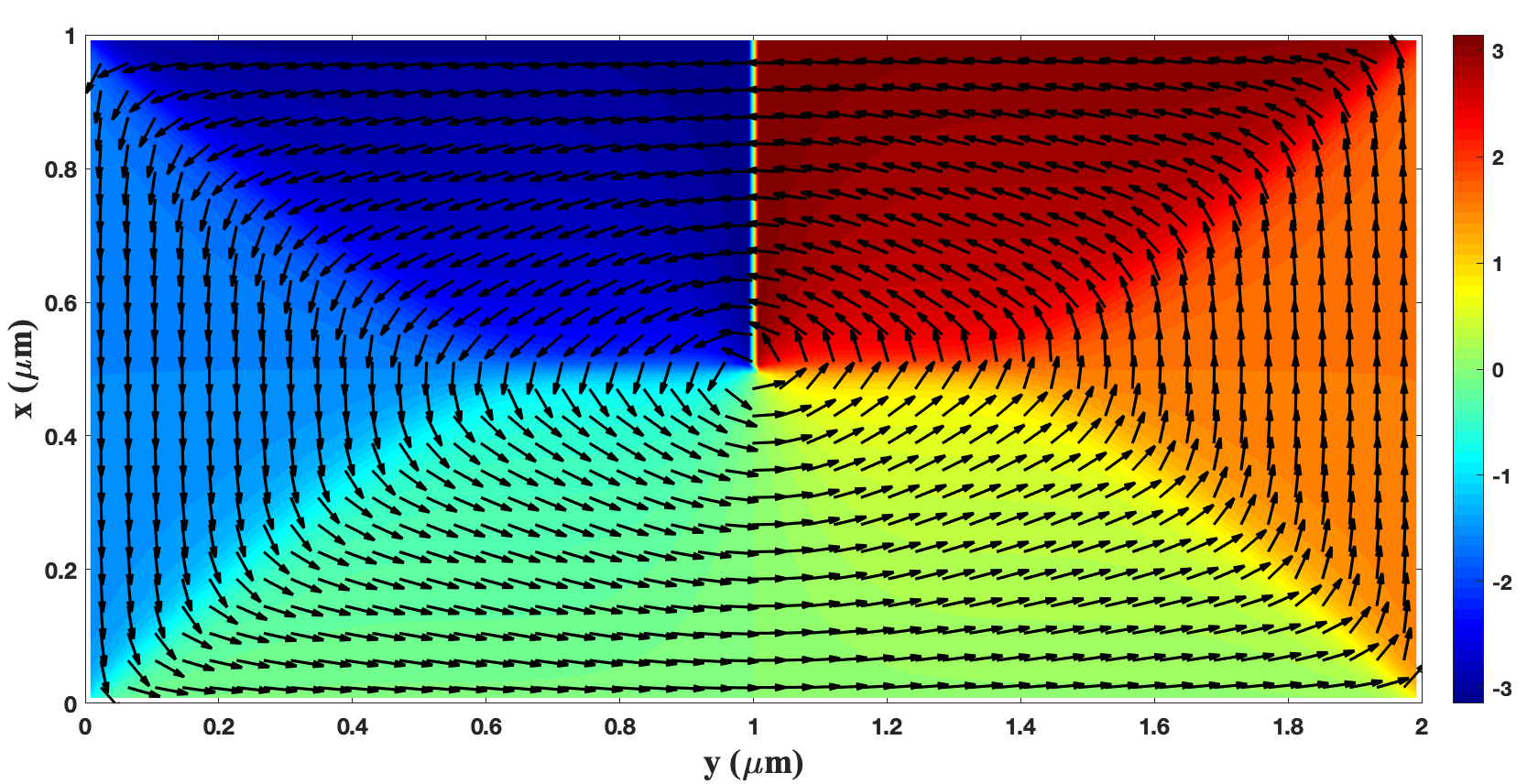}
\end{minipage}
\centering
\begin{minipage}[t]{0.49\textwidth}
\centering
\includegraphics[width=6cm]{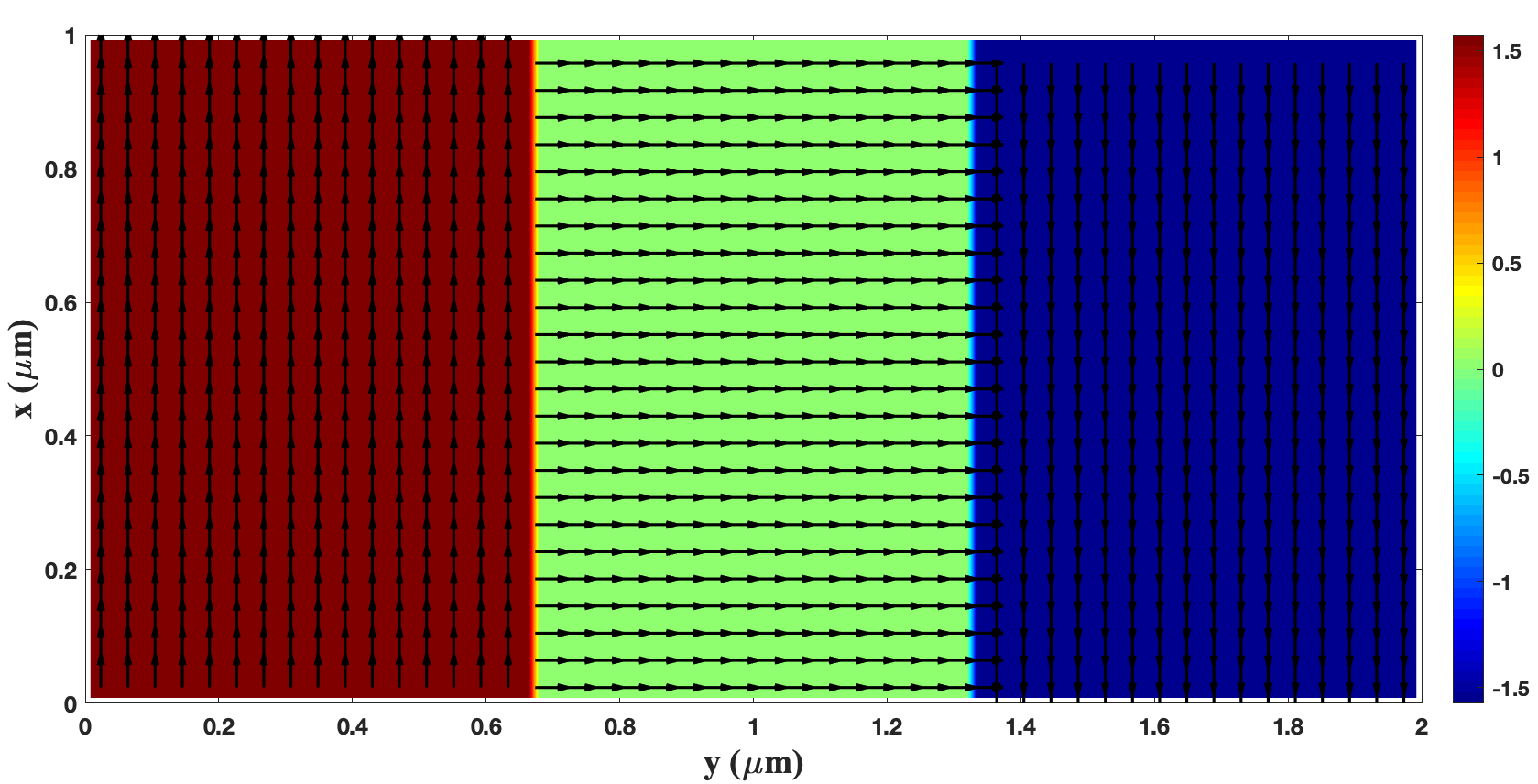}
\end{minipage}
\begin{minipage}[t]{0.49\textwidth}
\centering
\includegraphics[width=6cm]{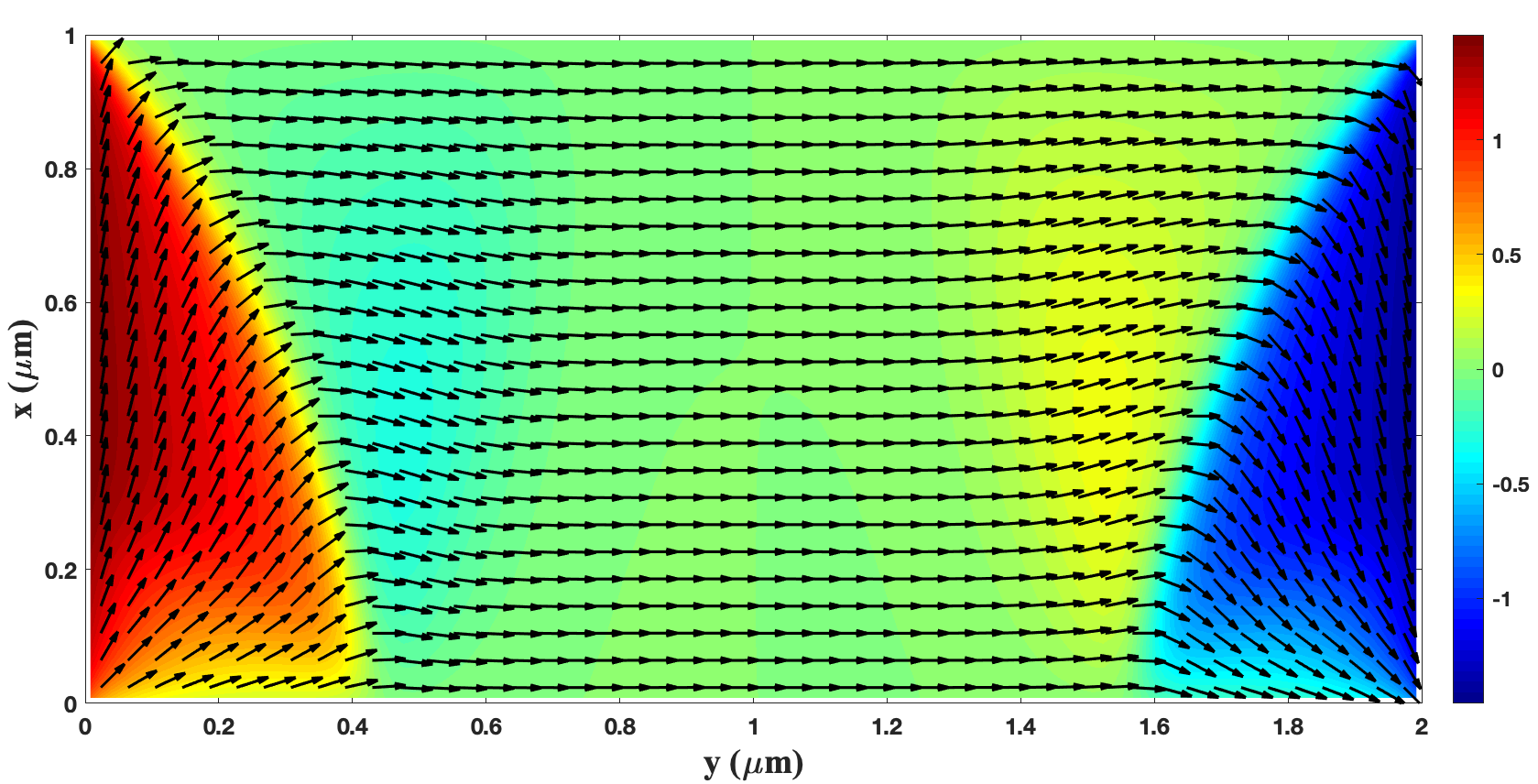}
\end{minipage}
\centering
\begin{minipage}[t]{0.49\textwidth}
\centering
\includegraphics[width=6cm]{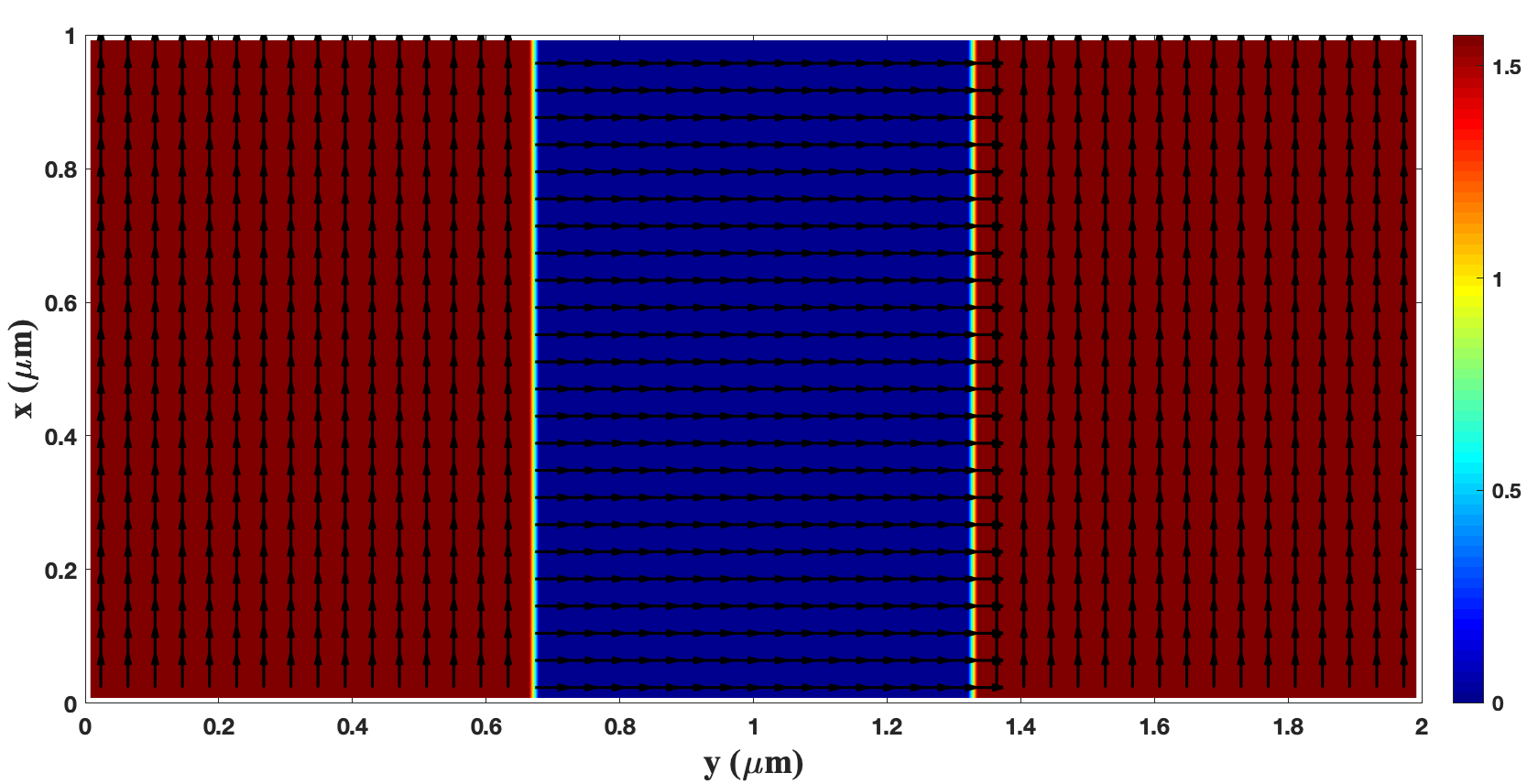}
\end{minipage}
\begin{minipage}[t]{0.49\textwidth}
\centering
\includegraphics[width=6cm]{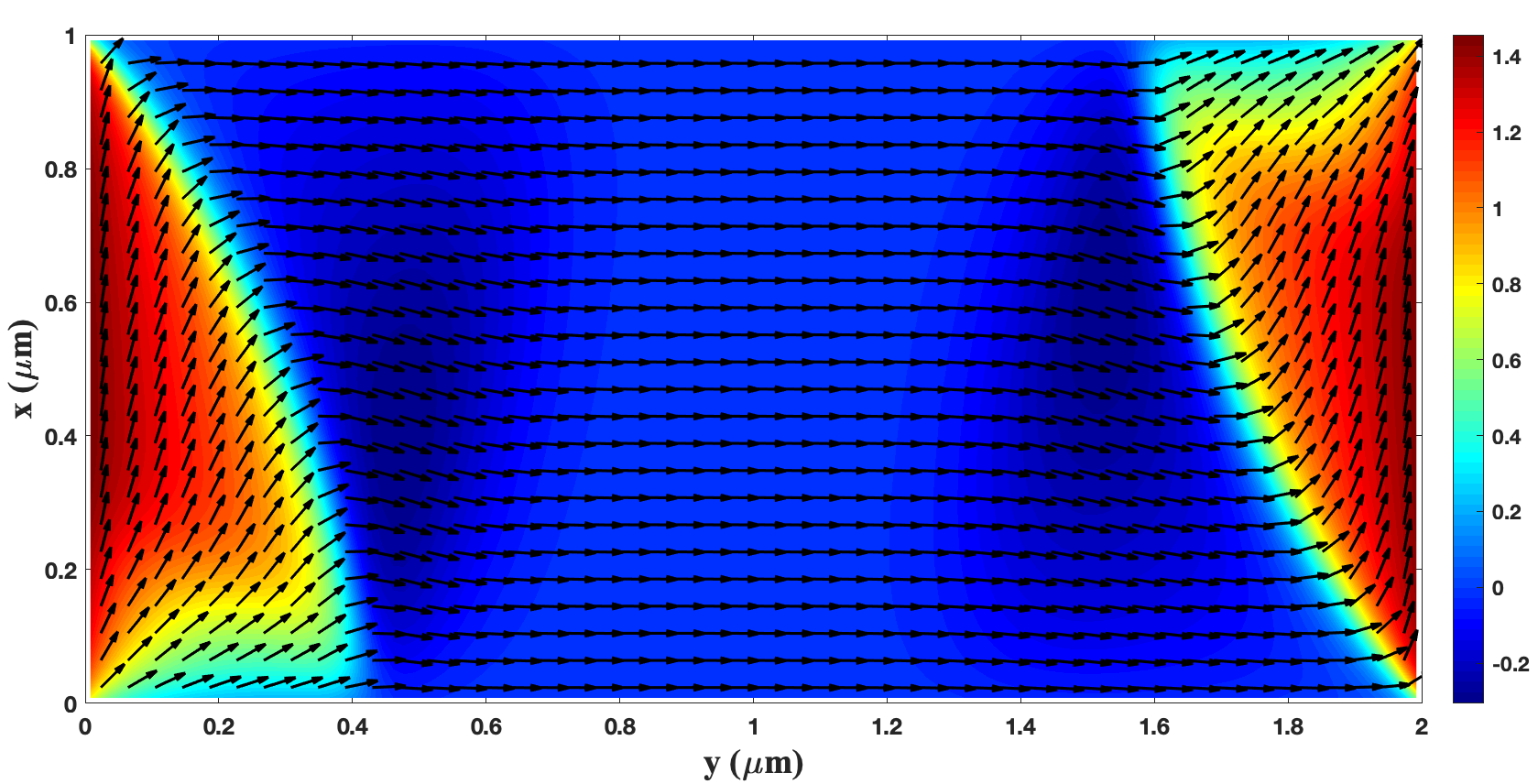}
\end{minipage}
\caption{Four equilibrium states simulated by the SICN scheme. The arrow denotes the first two components of the magnetization vector and the color denotes the angle between them. Top row: Diamond state; Second row: Landau state; Third row: C-state; Bottom row: S-state. Left column: Initial state; Right column: Equilibrium state.}\label{fig:state}
\end{figure}
For four different initial magnetization distributions, the time evolutionary curve of the system energy \eqref{LL-Energy} is plotted in \cref{fig:energy}, computed by the SICN scheme. A theoretical justification of the energy dissipation is not available, while such a dissipation is clearly observed in the numerical simulation. Landau state is found to be the most stable structure in this case.
\begin{figure}[hbt]
  \includegraphics[width=0.8\textwidth]{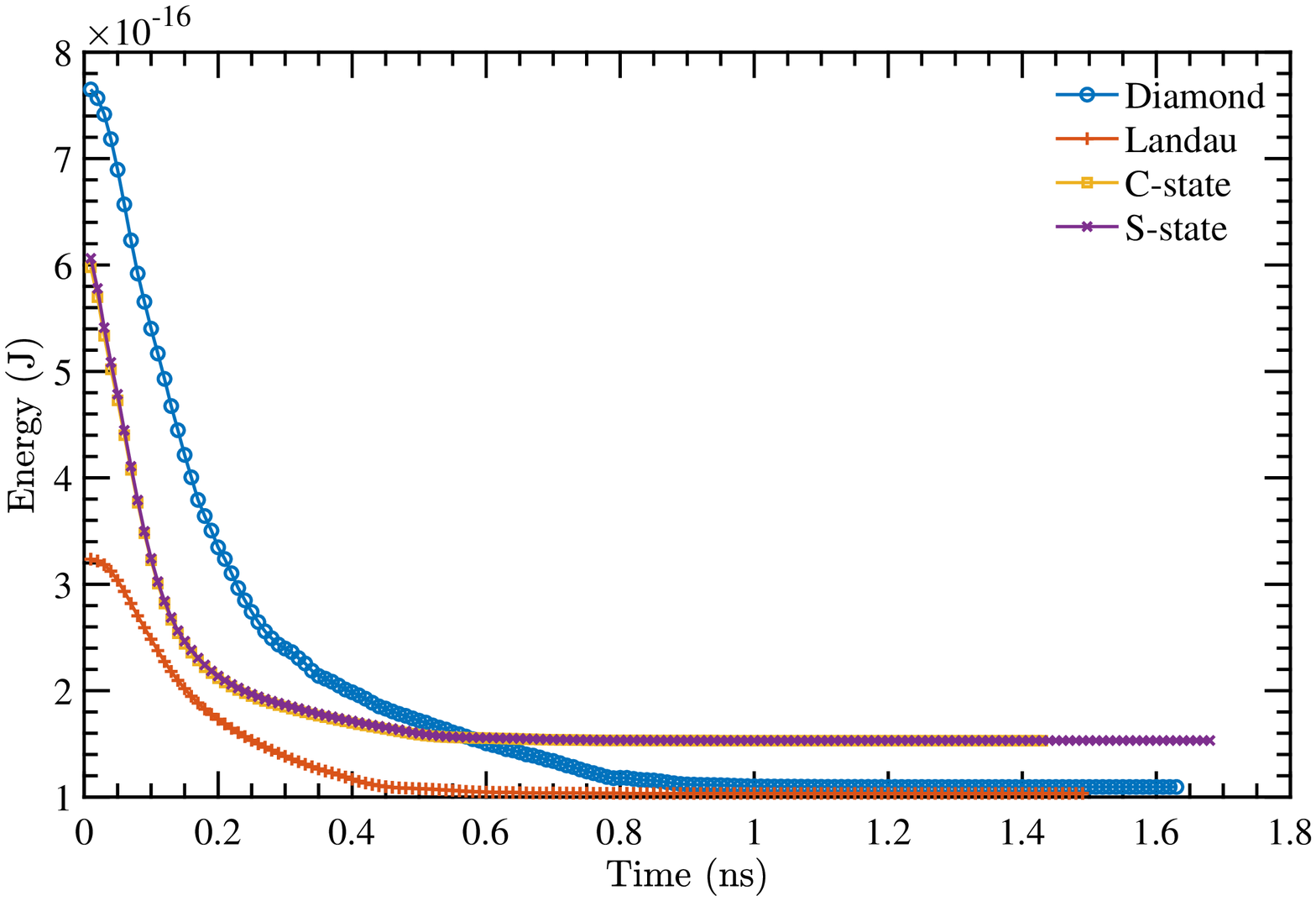}
  \caption{The time evolution of the system energy, computed by the SICN scheme. Energy dissipation is clearly observed in the numerical simulation.}
  \label{fig:energy}
\end{figure}

\subsection{Benchmark problem from NIST}

To check the practical performance of the SICN scheme in the realistic material, we simulate a standard problem proposed by the Micromagnetic Modeling Activity Group (muMag) from NIST \cite{NIST}. The hysteresis loop is simulated in the following way. A positive external field of strength $H_{0}=\mu_{0} H_{e}$ in the unit of $mT$ is applied. The magnetization is able to reach a steady-state. Once this steady-state is approached, the applied external field is reduced by a certain amount and the material sample is allowed to reach a steady state again. Repeat the process until the hysteresis system attains a negative external field of strength $H_{0}$. This process is then implemented in reverse, increasing the field in small steps until the initial applied external field is reached. Therefore, we are now able to plot the average magnetization at the steady-state as a function of the external field strength during the hysteresis loop. The stopping criterion for a steady state is that the relative change of the total energy is less than $10^{-9}$. The applied field is parallel to the $x$ axis. According to the available code \emph{mo96a} of the first standard problem from NIST, we set $ 100\times 50 \times 1 $ spatial resolution with mesh size $ 20 \times 20 \times 20 \ \mathrm{nm}^3  $ and the canting angle $ +1^\circ $ of applied field from the nominal axis. The initial state is uniform and $ [-50\;\mathrm{mT}, +50\;\mathrm{mT}] $ is split into 200 steps for both $x$-loop and $y$-loop. The material parameters and the temporal step size are the same as those in \cref{sec:states}.

Hysteresis loops generated by the code \emph{mo96a} are displayed in \cref{hloop}(a) and (b) when the applied field is approximately parallel to the $y$-(long) axis and the $x$-(short) axis, respectively. The average remanent magnetization in reduced units is given by $(-1.5120\times 10^{-1},8.6964\times 10^{-1},0)$ for the $y$-loop and $(1.5257\times 10^{-1},8.6870\times 10^{-1},0)$ for the $x$-loop. The coercive fields are 4.8871 mT in \cref{hloop}(a) and -2.5253 mT in \cref{hloop}(b). Hysteresis loops generated by the SICN scheme are presented in \cref{hloop}(c) and (d) when the applied field is approximately parallel to the long axis and the short axis, respectively. The average remanent magnetization in reduced units is given by $ (-1.5473\times 10^{-1},8.7251\times 10^{-1},1.9726\times 10^{-5}) $ for the $y$-loop and $ (1.5525\times 10^{-1},8.7201\times 10^{-1},2.1583\times 10^{-5}) $ for the $x$-loop. The coercive fields are $ 6.0837(\pm0.4)$ mT in \cref{hloop}(c) and $-2.6701 (\pm 0.4) $mT in \cref{hloop}(d). Therefore, we conclude that the results of the SICN scheme agree well with those of NIST both qualitatively and quantitatively.
\begin{figure}[htp]
\centering 
\subfigure[$ H_0 \parallel \text{y-axis},mo96a$]
{
\includegraphics[scale=0.2]{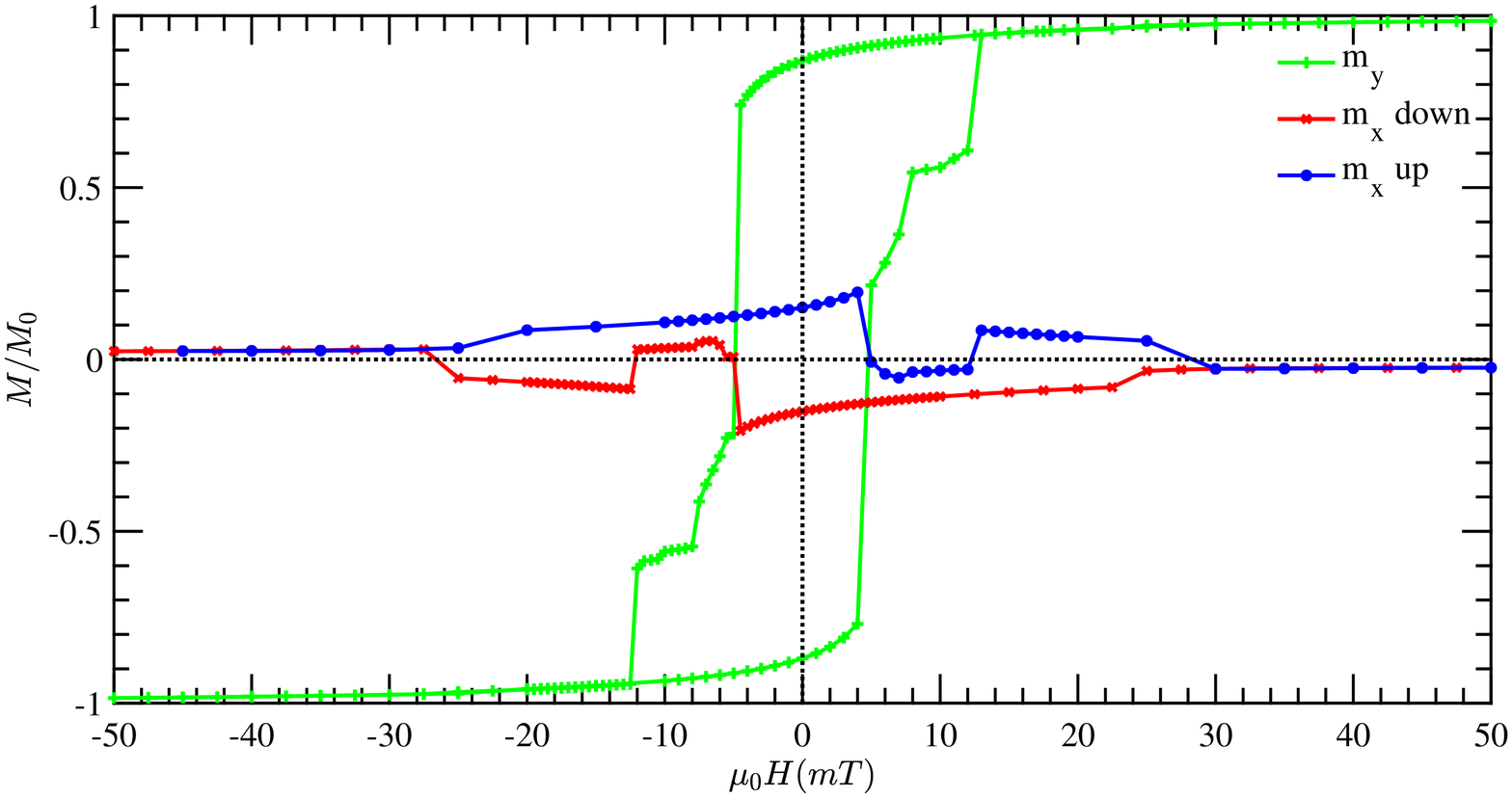}
}
\subfigure[$ H_0 \parallel \text{x-axis},mo96a$]
{ 
\includegraphics[scale=0.2]{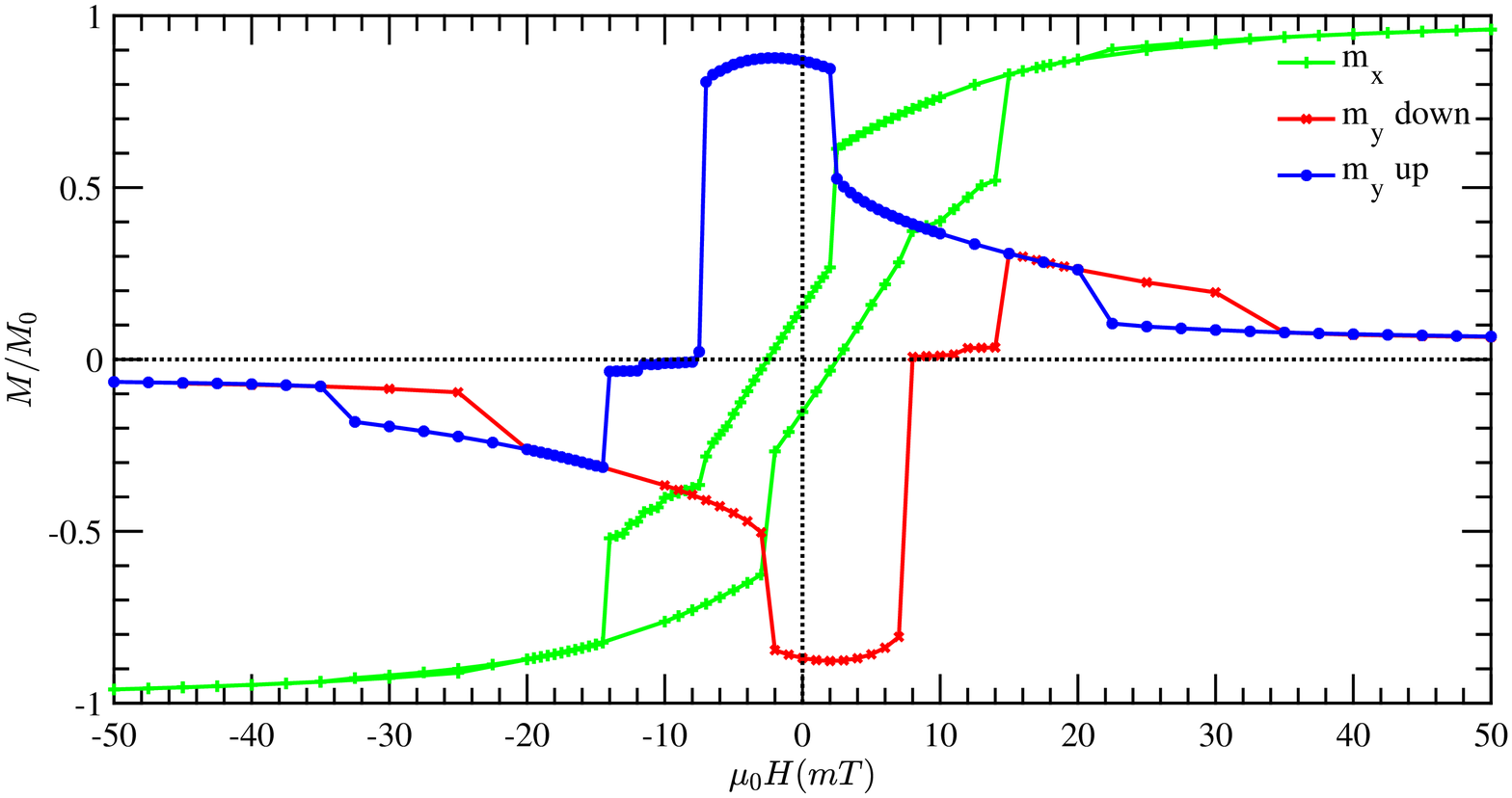}          
}
\subfigure[$ H_0 \parallel \text{y-axis},\text{SICN scheme}$]
{     
\includegraphics[scale=0.2]{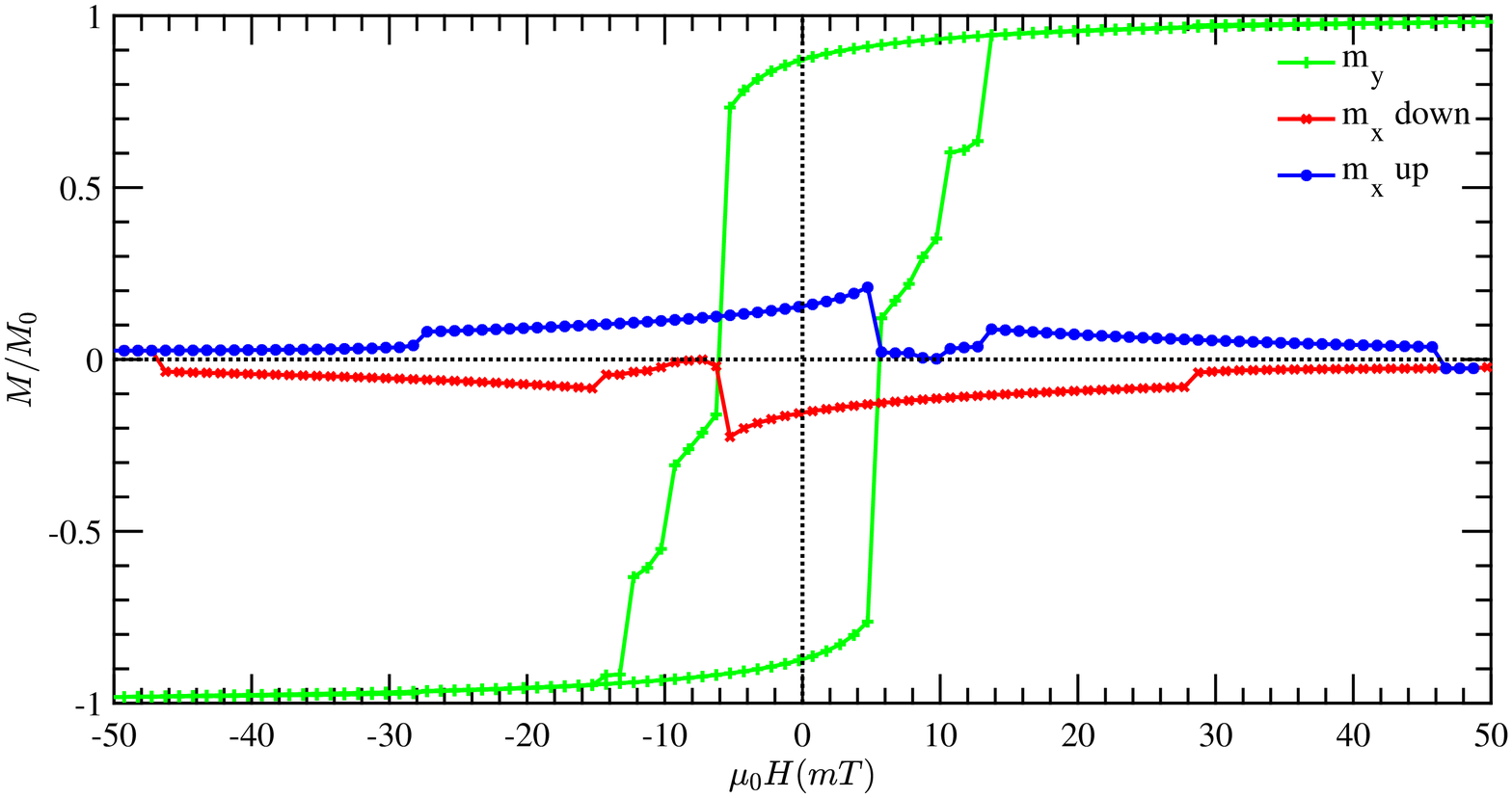}       
}
\subfigure[$ H_0 \parallel \text{x-axis},\text{SICN scheme}$]
{ 
\includegraphics[scale=0.2]{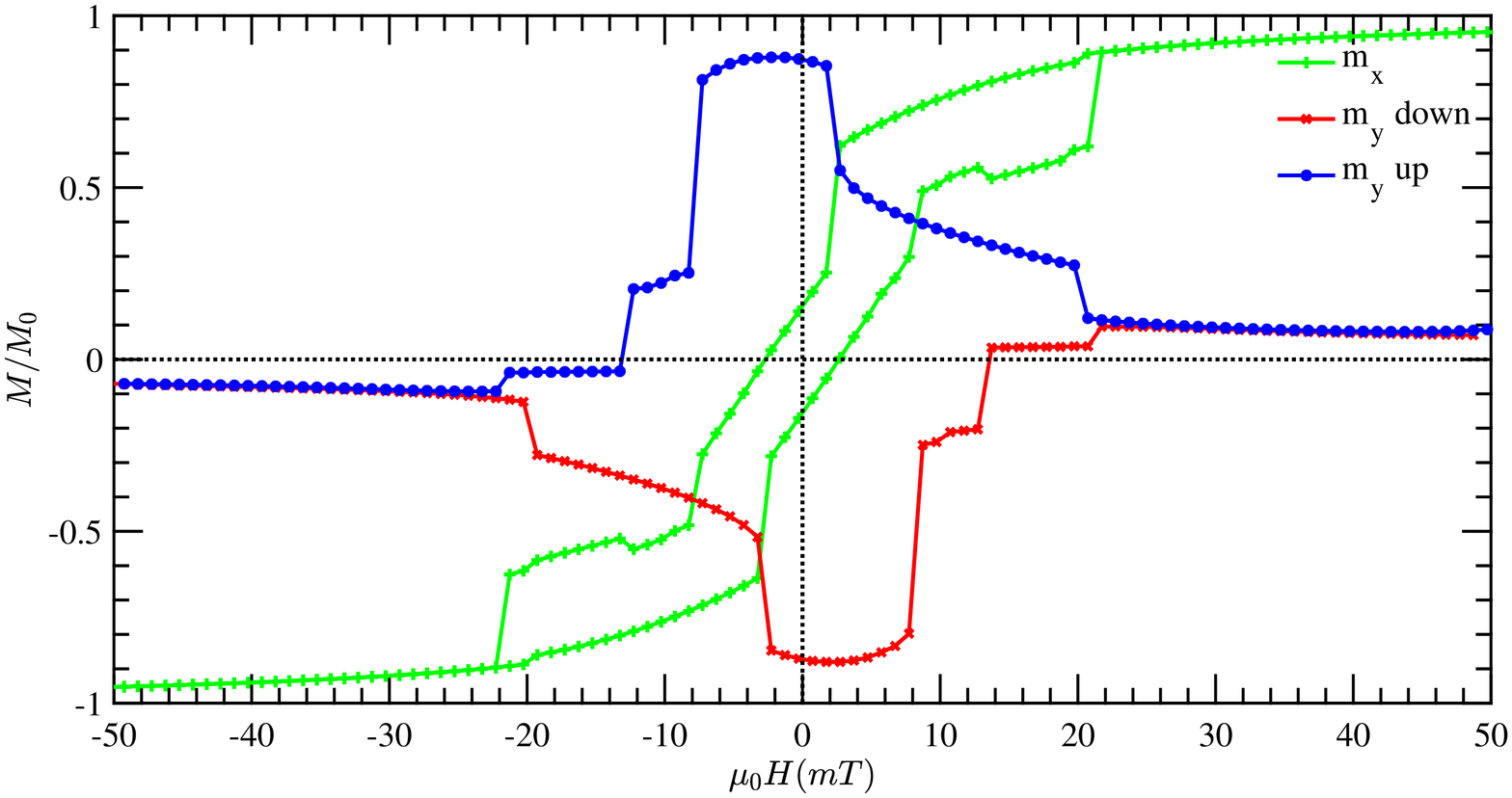}          
}
\caption{Hysteresis loops when $ \alpha=0.1 $ with the mesh size $ 20 \times 20 \times 20 \ \mathrm{nm}^3  $. The applied field is approximately parallel (canting angle $+1^\circ  $ ) to the $y$-axis (left column) and the $x$-axis (right column). Top row: \emph{mo96a}; Bottom row: SICN scheme.} 
\label{hloop}
\end{figure}

\section{Conclusions}\label{sec:conclusion}
In this work, we conduct a comprehensive study of the ICN scheme and the SICN scheme for the LLG equation. Theoretically, both schemes are second-order accurate in space and time. The unique solvability of nonlinear systems of equations in the ICN scheme requires a severe time step constraint, $k=\mathcal{O}(h^2)$, while the linear system of equations in the SICN scheme leads to an unconditionally unique solvability. Meanwhile, a much milder condition $k=\mathcal{O}(h)$ is needed to ensure the numerical convergence of the SICN scheme. Numerically, it is discovered that more than one solution may exist for the nonlinear system of equations in the ICN scheme. In terms of numerical performance, the SICN scheme reduces the CPU time by at least $50\%$, in comparison to the ICN scheme for the same accuracy requirement. Therefore, we strongly suggest that the semi-implicit scheme shall be used in micromagnetics simulations if both schemes are available.

\section*{Acknowledgements}
\thanks{This work was supported by Undergraduate Training Program for Innovation and Entrepreneurship, Soochow University Project 201910285013Z (Y. Sun), NSFC grants 11971021 (J. Chen) and 11501399 (R. Du), and NSF DMS-2012669 (C.~Wang).}

\bibliographystyle{plain}
\bibliography{LLG}

\end{document}